\pdfoutput=1

\documentclass[a4paper,UKenglish,cleveref, autoref, thm-restate]{lipics-v2021}

\title{A Circuit Approach to Constructing Blockchains on Blockchains}

\author{Ertem Nusret Tas}{Stanford University, USA}{nusret@stanford.edu}{https://orcid.org/0000-0001-6061-9700}{}

\author{David Tse}{Stanford University, USA}{dntse@stanford.edu}{https://orcid.org/0000-0003-1460-5900}{}

\author{Yifei Wang}{Stanford University, USA}{wangyf18@stanford.edu}{https://orcid.org/0000-0002-0364-0893}{}

\authorrunning{E.N. Tas, Y. Wang and D. Tse} %

\Copyright{Ertem Nusret Tas and Yifei Wang and David Tse} %



\category{} %

\acknowledgements{
We thank Dionysis Zindros for several insightful discussions on this project. This paper and the concurrent related work in which blockchains were analyzed as `virtual parties'~\cite{rollerblade} both came out of many fruitful discussions about blockchain composability among Ertem Nusret Tas, David Tse, Yifei Wang and Dionysis Zindros, when they were all at Stanford University. This research was funded by a Research Hub Collaboration agreement with Input Output Global Inc. Ertem Nusret Tas is supported by the Stanford Center for Blockchain Research.
}%

\nolinenumbers %

\EventEditors{John Q. Open and Joan R. Access}
\EventNoEds{2}
\EventLongTitle{42nd Conference on Very Important Topics (CVIT 2016)}
\EventShortTitle{CVIT 2016}
\EventAcronym{CVIT}
\EventYear{2016}
\EventDate{December 24--27, 2016}
\EventLocation{Little Whinging, United Kingdom}
\EventLogo{}
\SeriesVolume{42}
\ArticleNo{23}

\usepackage{preamble}

\hideLIPIcs

\begin{document}

\maketitle

\begin{abstract}
Recent years have witnessed an explosion of blockchains, each with an open ledger that anyone can read from and write to. In this multi-chain world, an important question emerges: how can we build a more secure overlay blockchain by reading from and writing to a given set of blockchains? Drawing an analogy with switching circuits, we approach the problem by defining two basic compositional operations between blockchains, serial and \lvl compositions, and use these operations as building blocks to construct general overlay blockchains. Under the partially synchronous setting, we have the following results: 1) the serial composition, between two certificate-producing blockchains, yields an overlay blockchain that is safe if at least one of the two underlay blockchains is safe and that is live if both of them are live; 2) the \lvl composition between three blockchains, akin to parallel composition of switching circuits, yields an overlay blockchain that is safe if all underlay blockchains are safe and that is live if over half of them are live; 3) repeated composition of these two basic operations can yield all possible tradeoffs of safety and liveness for an overlay blockchain built on an arbitrary number of underlay chains.
The results are also extended to the synchronous setting\let\thefootnote\relax\footnote{The authors are listed alphabetically.}.
\end{abstract}

\section{Introduction}
\label{sec:introduction}

\subsection{Background}

Bitcoin, invented by Nakamoto in 2008 \cite{bitcoin}, is the first blockchain with a public ledger, which anybody can read from and write arbitrary data.
Since then, there has been a proliferation of such blockchains. 
Each of them is a consensus protocol run by its own set of validators.
Together, these blockchains form a {\em multi-chain world}, communicating with each other through bridging protocols which read from and write to the blockchains.

As a consensus protocol, a fundamental property of a blockchain is its {\em security}: a blockchain is secure if the ledger it provides is safe and live. 
The security is supported by the blockchain's set of validators. 
In a multi-chain world, a natural question arises: given a set of existing blockchains, how to build a more secure protocol, an {\em overlay} blockchain, by only reading from and writing to the ledgers of the individual {\em underlay} blockchains? 
In other words, how to build a blockchain on blockchains?

This problem has received attention recently, and there have been two main approaches to this problem in the literature. 
The first approach is {\em interchain timestamping}.
In the context of two blockchains, data on one blockchain is timestamped to another blockchain, and a more secure ledger is obtained by reading the ledger of the first chain using the timestamps on the second chain to resolve forks.
An interchain timestamping protocol was proposed in \cite{bitcoin-timestamp} to allow a Proof-of-Stake (PoS) chain to borrow security from Bitcoin. 
In that work, Bitcoin is assumed to be secure and the problem was to determine the optimal security properties that can be achieved by the PoS chain. 
A more symmetric formulation is considered in \cite{tas2023interchain}, where none of the individual chains is assumed to be secure. 
Moreover, the interchain timetamping protocol is extended to more than $2$ chains, where the timestamping proceeds in a {\em sequential} manner, where the chains are ordered and the first chain timestamps to a second chain which timestamps to a third chain, etc. 
The main security result in \cite{tas2023interchain} is that the overlay blockchain is safe if at least one of the underlay blockchains is safe, and is live if all of the underlay blockchains are live.

In the second approach, an analogy is drawn between the multiple blockchains and the multiple validators in a blockchain, 
and an overlay blockchain is built by running a consensus protocol on top of the underlay blockchains by treating them as validators. 
This idea was first sketched out in Recursive Tendermint \cite{recursive-tendermint} in the context of the Cosmos ecosystem, consisting of numerous application specific blockchains each running the Tendermint consensus protocol \cite{tendermint}. 
Recently, this idea was made more precise and concrete by Trustboost \cite{trustboost}, where the validator role of each underlay blockchain is instantiated by a specialized smart contract. 
These simulated validators send messages between the underlay blockchains via a cross-chain communication protocol (CCC) to implement a variant~\cite{iths} of the Hotstuff consensus protocol \cite{yin2018hotstuff}. 
The main security result in \cite{trustboost} is that, in a partially synchronous network,  the overlay blockchain is secure (safe and live) if more than  $2/3$ of the underlay blockchains are secure.

\subsection{Problem Motivation}

Even though interchain timestamping and Trustboost both propose a construction of blockchains on blockchains, their security statements are quite different in nature. 
First, the conditions for safety and liveness are separate for the interchain timestamping protocol, while they are coupled in Trustboost.
Since loss of safety and loss of liveness may have different impacts on a blockchain, separating out when safety and liveness are achieved is useful. 

Second, when the safety condition of the overlay blockchain depends {\em only} on the safety of the underlay blockchains, one can immediately infer the {\em accountable safety}~\cite{casper} (also known as the forensics property~\cite{forensics}) of the overlay blockchain in terms of the accountable safety of the underlay chains.
Accountable safety states that if the adversary controls a large fraction of the validators and causes a safety violation, all protocol observers can irrefutably identify the adversarial validators responsible for the safety violation.
It is thus a strengthening of the traditional safety guarantees of consensus protocols.
When the overlay blockchain's safety depends only on the underlay chains' safety, a safety violation on the overlay would imply safety violations on (some of) the underlays.
Therefore, if the underlay chains satisfy accountable safety, when the overlay's safety is violated, all protocol observers would identify the responsible adversarial validators of the underlay chains, implying accountable safety for the overlay blockchain.

Whereas there are protocols satisfying accountable safety~\cite{snapandchat,forensics}, adversarial validators responsible for liveness violations cannot be held accountable in the same sense as accountable safety~\cite{btc-pos}.
Thus, to infer the accountable safety of an overlay blockchain, its safety should depend \emph{only} on the safety but not the liveness of the underlay blockchains.
Indeed, if the overlay blockchain loses safety due to liveness violations on the underlay blockchains, it might not be possible to identify the responsible adversarial validators. 

Finally, separating safety and liveness of the overlay blockchain and characterizing their dependence on the safety and liveness of the underlay chains enables achieving greater resilience than when security is based on the number of secure underlay chains, with safety and liveness \emph{coupled}.
For illustration, Trustboost~\cite{trustboost} shows security of the overlay blockchain only when over $2/3$ of the underlay chains are secure, \ie, \emph{both} safe and live.
Note that this is optimal if the overlay chain's security were to be based on the number of secure underlay chains.
However, by separating safety and liveness, we can achieve safety for the overlay blockchain if over $2/3$ of the underlay chains are safe, and liveness if over $2/3$ of the underlay chains are live, where the sets of safe and live underlay chains need not be the \emph{same}.
This implies that the overlay blockchain can be secure, even when up to $2/3$ of the underlay chains are not both safe and live, \ie, secure!
While this may sound puzzling, there are no hidden tricks at play here.
Indeed, any two quorums of underlay chains required for the liveness of the overlay blockchain must intersect at an underlay chain whose safety is required for the overlay's safety.
In contrast, using our notation, Trustboost would require both liveness quorums to intersect with a safety quorum at over $2/3$ of the underlay chains.

Interchain timestamping protocols provide an inspiration for security statements separating out safety and liveness, but they only achieve one particular tradeoff between safety and liveness: they favor safety strongly over liveness.
This is because safety of the overlay blockchain requires only one of the underlay blockchains to be safe, while liveness of the overlay blockchain requires all of the underlay blockchains to be live.
Therefore, two natural questions arise: 1) What are all the tradeoffs between safety and liveness which can be achieved? 2) How can we construct overlay blockchains that can achieve all the tradeoffs? 
The main contributions of this paper are to answer these two questions.

\subsection{Security Theorems}

Consider overlay blockchains instantiated with $k$ underlay chains (\cf Section~\ref{sec:prelim} for a formal definition of overlay blockchains). 
We say a tuple $(k, s, l)$ is achievable if one can construct an overlay blockchain such that

\noindent
a. If $s$ or more underlay blockchains are safe, the overlay blockchain is safe.

\noindent
b. If $l$ or more underlay blockchains are live with constant latency after the global stabilization time (GST), the overlay blockchain is live with constant latency after GST. 

\noindent
Going forward, when referring to the liveness of a blockchain, we mean liveness with constant latency after GST.

We identify all achievable tuples and provide a protocol achieving them (Fig.~\ref{fig:achieve}). 
\begin{theorem}
\label{thm:informal}
Consider the partially synchronous setting. 
For any integers $k\geq1$, $l$ and $s$ such that $\flr{k/2}+1\leq l\leq k$ and $s\geq 2(k-l)+1$, the tuple $(k,s,l)$ is achievable.
\end{theorem}
In particular, the tuple $(k,\clr{\frac{2k}{3}}, \clr{\frac{2k}{3}})$ is achievable, \ie, there is an overlay blockchain that is safe if more than $2/3$ of the underlay chains are safe, and is live if more than $2/3$ of the underlay chains are live. 
This implies that the overlay is safe and live if more than $2/3$ of the underlay chains are safe and more than $2/3$ of the underlay chains are live. 
Note that this is a strictly stronger security guarantee than Trustboost, which is guaranteed to be safe and live if more than $2/3$ of the underlay chains are both safe and live; \ie, the {\em same} chains need to be safe and live in this latter statement. 
Moreover, Theorem \ref{thm:informal} includes also asymmetric operating points where $s \neq \ell$. 

The next theorem gives a matching impossibility result.
\begin{theorem}[Informal, Theorem~\ref{lem:converse-symmetric-psync}]
\label{thm:impos}
Consider the partially synchronous network.
For any integers $k\geq1$, $l$ and $s$ such that $\flr{k/2}+1\leq l\leq k$ and $s<2(k-l)+1$, no protocol can satisfy the following properties simultaneously:

\noindent
a. If $s$ underlay blockchains are safe and \emph{all underlay blockchains are live}, the overlay blockchain is safe.

\noindent
b. If $l$ underlay blockchains are live and \emph{all underlay blockchains are safe}, the overlay blockchain is live.

The same result holds for any integers $k \geq 1$ and $l \leq k/2$.
\end{theorem}
Theorem~\ref{thm:impos} shows the optimality of the result in Theorem~\ref{thm:informal} in a strong sense:  even if we allow the safety or liveness of the overlay blockchain to depend on \emph{both} safety and liveness of the underlay chains (Fig.~\ref{fig:achieve}), the security guarantee of the overlay blockchain cannot be improved.
In other words, under partial synchrony, liveness of the underlay chains have no effect on the safety of the overlay blockchain.
Theorems~\ref{thm:informal} and~\ref{thm:impos} are proven in Sections~\ref{sec:circuit-symmetric-psync} and~\ref{sec:converse-symmetric-psync} respectively.

We also characterize the security properties achievable in the \emph{synchronous} network.
\begin{theorem}[Informal, Theorems~\ref{lem:circuit-symmetric-sync} and~\ref{lem:converse-symmetric-sync}]
\label{thm:informal-sync}
Consider the synchronous network.
For any integers $k\geq1$, $l$, $s$ and $b$, one can construct an overlay blockchain as described below if and only if $\flr{k/2}+1\leq l\leq k$, $s \geq 2(k-l)+1$, and $b \geq k-l+1$:

\noindent
a. If $s$ underlay blockchains are safe, \emph{or} $b$ underlay blockchains are both safe and live, the overlay blockchain is safe.

\noindent
b. If $l$ underlay blockchains are live, the overlay blockchain is live.
\end{theorem}

Theorem~\ref{thm:informal-sync} shows that under synchrony, unlike partial synchrony, the overlay blockchain has better safety guarantees when the underlay chains are \emph{both} safe and live.
On the other hand, Theorem~\ref{thm:informal-sync} implies that if we require the safety (liveness) of the overlay blockchain to depend only on the number of safe (live) underlay chains (\ie, restrict $b$ to be $0$), we cannot achieve any better resilience under synchrony compared to partial synchrony.
Security in the synchronous network is discussed further in Sections~\ref{sec:synchrony} and~\ref{sec:converse-sync}.

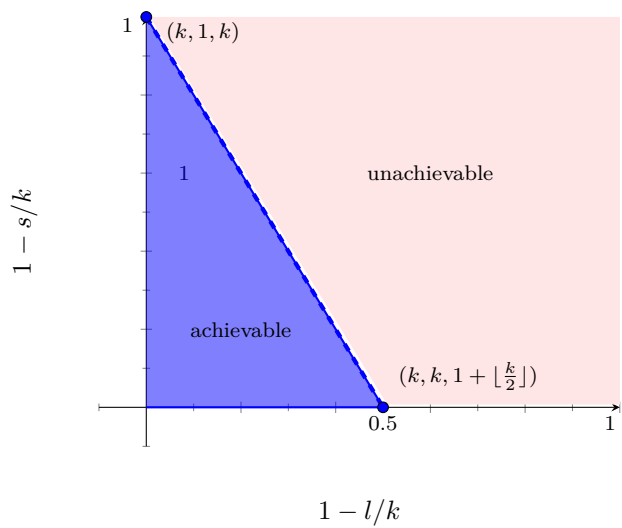
\begin{figure}
\centering
\begin{minipage}{0.45\textwidth}
\begin{center}
\begin{tikzpicture}
  \begin{axis}[ymin=-0.5,ymax=5,xmax=5,xmin=-0.5,xticklabel=\empty,yticklabel=\empty,
               minor tick num=1,axis lines = middle,xlabel=$1-l/k$,ylabel=$1-s/k$,x label style={at={(axis description cs:0.5,-0.1)},anchor=north},
    y label style={at={(axis description cs:-0.1,.5)},rotate=90,anchor=south}]
\addplot [-,>={Latex[round]},domain=0:2.5,samples=2,ultra thick,color=blue,dashed] {5-2*x};
\addplot [only marks,mark=*,fill=blue] coordinates { (2.5,0) };
\addplot [only marks,mark=*,fill=blue] coordinates { (0,5) };
\node[font=\fontsize{8}{8}\selectfont] at (axis cs:4.9,-.2){1}; 
\node[font=\fontsize{8}{8}\selectfont] at (axis cs:2.5,-.2){0.5}; 
\node[font=\fontsize{8}{8}\selectfont] at (axis cs:-.2,4.9){1}; 
\node[font=\fontsize{8}{8}\selectfont] at (axis cs:.4,3){1};
\addplot [thick,color=mylightred,fill=mylightred, 
                    fill opacity=1]coordinates {
            (0.05,5)
            (2.55,0.05)
            (5,0.05)
            (5,5)};
\addplot [thick,color=blue,fill=blue, 
                    fill opacity=0.5]coordinates {
            (0,5)
            (2.5,0)
            (0,0)};
\node[font=\fontsize{8}{8}\selectfont] at (axis cs:3,3){unachievable}; 
\node[font=\fontsize{8}{8}\selectfont] at (axis cs:1,1){achievable}; 
\node[font=\fontsize{8}{8}\selectfont] at (axis cs:3.4,.4){$(k,k,1+\lfloor\frac{k}{2}\rfloor)$};
\node[font=\fontsize{8}{8}\selectfont] at (axis cs:.6,4.8){$(k,1,k)$};
  \end{axis}

\end{tikzpicture}
\end{center}
\end{minipage}
\caption{Region of safety-liveness guarantee. The integer grids in the blue area consists of all points which are achievable, while the integer points in the red area are not achievable under partial synchrony. We highlight the two extreme achievable tuples $(k,1,k)$ and $(k,k,1+\lfloor \frac{k}{2} \rfloor)$. }\label{fig:achieve}
\end{figure}

\subsection{Construction via Blockchain Circuits}

We now give insight into our methods using the example of the $(k,s,l)$ tuples under partial synchrony. 
For $k=2$, the only achievable tuple in Theorem \ref{thm:informal} is $(2,1,2)$, which can be achieved by timestamping.
For $k=3$, we have $(3,1,3)$ and $(3,3,2)$ as achievable tuples.
$(3,1,3)$ can be achieved by sequential interchain timestamping across $3$ chains. 
This is the strongly safety favoring overlay blockchain (extremal of the tradeoff in Figure \ref{fig:achieve}). 
$(3,3,2)$ represents a liveness-favoring overlay blockchain: it is safe if {\em all} $3$ underlay blockchains are safe, and is live if at least $2$ of the $3$ underlay blockchains are live. 
No existing construction is known to achieve this operating point.
Our solution to achieve all tuples in Theorem \ref{thm:informal} consists of two steps and described in Sections~\ref{sec:protocol-primitives} and~\ref{sec:partial-synchrony}:

\noindent
1. We provide a construction that achieves $(3,3,2)$ by drawing an analogy to Omission-Fault Tolerant (OFT) protocols, where validators only commit omission faults (analogous to loss of liveness of a safe underlay blockchain) but no Byzantine faults.
    
\noindent    
2. We show that by repeatedly composing the $(2,1,2)$ and $(3,3,2)$ solutions, we can build overlay blockchains that achieve any tuple in Theorem \ref{thm:informal}. 

As an inspiration to our approach, we can draw an analogy to switching circuit design in Claude E. Shannon's masters' thesis~\cite{shannon1938symbolic} (Table~\ref{tab:comp}). 
In this spectacular masters' thesis, Shannon used serial and parallel composition of switches to create an OR and an AND gate respectively, and then use these gates as building blocks to create more complex circuits which can be designed using Boolean algebra. 
Drawing the analogy, the timestamping solution to $(2,1,2)$ can be viewed as a {\em serial composition} of two blockchains, and the OFT solution to $(3,3,2)$, called the \lvl composition due to the use of three blockchains, can be viewed as a {\em parallel composition} of three blockchains for partial synchrony (Curiously, unlike switching circuits, no parallel composition of $2$ blockchains can exist \emph{under partial synchrony}, as ruled out by Theorem~\ref{thm:impos}. See Section~\ref{sec:view} for more discussion.)

Our serial and parallel compositions require the composed underlay blockchains to satisfy certain properties (\eg, hosting smart contracts) outlined in Section~\ref{sec:protocol-primitives}.
These properties are satisfied by blockchains that support general-purpose smart contracts (\eg, EVM in Ethereum) and run on PBFT-style consensus protocols~\cite{pbft} such as Tendermint~\cite{tendermint}.
In this context, our circuit compositions can be readily implemented by Cosmos blockchains~\cite{cosmos} that support CosmWasm smart contracts and run Tendermint as their consensus protocols.

\begin{table}
\centering
\resizebox{\linewidth}{!}{
   \begin{tabular}{c|c|c}
          & Switching Circuits & Blockchain Circuits  \\
          \hline
          \hline
          Goal & Computation & Security \\
          \hline
          Basic components & switches & blockchains\\
       \hline
       Component state  & $X \in \{0,1\}$ & $X = (S,L) \in \{0,1\}^2$\\
      & $X=1$ iff switch is open  & $S=1$ iff chain is safe\\
      
        &  & $L=1$ iff chain is live\\
       \hline
       Serial composition & $X_1 + X_2 = X_1$ OR $X_2$ & $X_1 + X_2 = (S_1+S_2, L_1L_2)$\\
       \hline
       Parallel composition & $X_1X_2 = X_1$ AND $X_2$ & $X_1X_2X_3 = (S_1S_2S_3, L_1L_2+L_2L_3+L_3L_1)$\\
       \hline
       Syntheis  & Boolean formulas & generalized quorum systems\\
       \hline
       Completeness & All truth table assignments & All achievable compositions
    \end{tabular}}
    \vspace{5pt}
     \caption{Comparison between switching circuits and blockchain circuits. We note that the parallel composition for blockchain circuits is more complicated than $X_1X_2=(S_1S_2,L_1+L_2)$, which would have been the natural analogue of the parallel composition. However, such a composition is impossible to achieve (Section \ref{sec:view}).}\label{tab:comp}
\end{table}

\subsection{Outline}
Our paper is organized as follows. 
Related works are summarized in Section~\ref{sec:related}. 
We present preliminary definitions in Section~\ref{sec:prelim}.
We describe the serial composition for achieving the tuple $(2,1,2)$ and the \lvl composition for achieving $(3,3,2)$
in Section~\ref{sec:protocol-primitives}.
Using them as \emph{gates}, we build circuit compositions achieving all possible security properties under partial synchrony and synchrony in Sections~\ref{sec:partial-synchrony} and~\ref{sec:synchrony}.
The converse results for unachievable properties under partial synchrony and synchrony are in Sections~\ref{sec:converse-psync} and~\ref{sec:converse-sync}.
Section~\ref{sec:efficiency} investigates scalability of large circuits based on serial and \lvl compositions.

\section{Related Works}
\label{sec:related}

\smallskip
\noindent
\textbf{Timestamping}. A timestamping protocol allows a \emph{consumer} chain to obtain timestamps for its blocks by checkpointing \cite{karakostas-checkpointing,gipp15a,bitcoin-timestamp,btc-pos,tas2023interchain} them on a \emph{provider} chain; so that in case there is a fork in the consumer chain, the fork can be resolved by choosing the one with the earlier timestamp (other uses of timestamping include reducing the latency of Nakamoto consensus~\cite{ledger-combiners}).
The provider chain is thus used as a timestamping server that provides security to the consumer chain.
Examples of timestamping protocols include Polygon \cite{polygon} checkpointing onto Ethereum, Stacks \cite{stacks} and Pikachu \cite{pikachu} checkpointing to PoW Ethereum and Babylon \cite{btc-pos} checkpointing to Bitcoin. %
Authors of \cite{tas2023interchain} design an interchain timestamping protocol to achieve \emph{mesh security} \cite{mesh-sec-github,sunny-mesh}, in which Cosmos zones provide and consume security to and from each other in a mesh architecture.
The protocol strongly favors safety over liveness and cannot achieve all possible security properties.

\smallskip
\noindent
\textbf{Trustboost.} Trustboost \cite{trustboost} proposes a family of protocols where multiple constituent blockchains
interact to produce a combined ledger with boosted trust.
Each blockchain runs a smart contract that emulates a validator of an overlay consensus protocol, Information Theoretic HotStuff (IT-HS) \cite{iths}, that outputs the ledger with boosted trust.
As long as over two-thirds of the blockchains are secure (safe and live), Trustboost satisfies security; thus its security guarantees are implied by our circuit construction.
Trustboost is implemented using Cosmos zones as the underlay blockchains and the inter-blockchain communication protocol (IBC) as the method of cross-chain communication; so that the emulated validators can exchange messages.
In this paper, we separate the safety/liveness conditions of the component blockchains for achieving safety/liveness guarantees of the interchain circuit construction.
Trustboost does not make any claims when the number of chains $k\leq 3$ or when a chain loses just one of its security properties (either safety or liveness), while our blockchain circuit approach covers all possible choices of achievable $(k,s,l)$ tuples, especially the two basic cases $(2,1,2)$ and $(3,3,2)$.
Trustboost also relies on external bots/scripts to notify the constituent blockchains about the overlay protocol's timeouts, whereas our approach does not use any external parties beyond the underlay blockchains.

As one can trade-off the safety and liveness resilience of HotStuff by tuning its quorum size, a natural question is if a similar trade-off for $(k,s,l)$ points can be achieved for Trustboost by tuning the quorum size of its overlay protocol (IT-HS).
However, to achieve these points, the overlay protocol must ensure liveness as long as $l$ blockchains are live, without requiring their safety.
This necessitates changing the overlay protocol to prioritize liveness in the case of safety violations\footnote{For instance, HotStuff must relax the liveness rule of the $\textsc{SafeNode}$ function to return true as long as the view number of the $\mathsf{prepareQC}$ is larger than or \emph{equal to} the locked view, which is different from the current specification in \cite{yin2018hotstuff}.}.
Then, a new security analysis is needed for the modified overlay protocol so that Trustboost can continue to leverage its security.
In contrast, the \lvl composition of our circuits builds on a liveness-favoring protocol as is (\cf Section~\ref{sec:view}).

\smallskip
\noindent
\textbf{Cross-staking.} 
Cross-staking was proposed as a technique to enhance the security of the Cosmos blockchains (zones) in the context of mesh security.
A consumer zone allows validators of a provider zone to stake their tokens on the consumer zone via IBC and validate the consumer chain.
However, this requires validators to run full nodes of multiple blockchains, thus resulting in a large overhead on that of interchain protocols and our blockchain circuit approach, where the validators of the constituent blockchains only run light clients.

\smallskip
\noindent
\textbf{Thunderella.}
Thunderella~\cite{thunderella} is a SMR consensus protocol, composed of an asynchronous, quorum-based protocol and a synchronous, longest chain based protocol.
The synchronous protocol ensures that Thunderella satisfies security, albeit with latency $O(\Delta)$, at all times with $1/2$ resilience under the $\Delta$-synchronous sleepy network model~\cite{sleepy}, whereas the asynchronous path helps achieve fast progress with latency dependent only on the actual network delay $\delta$, if over $3/4$ of the validators are honest and awake.
Thus, its goal is to support different latency regimes under different assumptions by having the validators execute two protocols, rather than to improve security by combining different chains in a black-box manner (\cf interchain consensus protocols, Section~\ref{sec:prelim}).

\smallskip
\noindent
\textbf{Robust Combines.}
Our approach of combining existing underlay chains to design a more secure overlay protocol is conceptually related to cryptographic combiners~\cite{Her05,HKN+05}, which combine many instances of a cryptographic primitive to obtain a more secure candidate for the same primitive.
The output satisfies correctness and security, if these properties are guaranteed for at least one of the original candidates.
In contrast, our circuit composition decouples safety and liveness and analyzes the dependence of the overlay protocols' safety and liveness \emph{separately} on the same properties of the underlay chains.

\section{Preliminaries}
\label{sec:prelim}
In this section, we introduce several preliminary definitions. 
We use $[k]$ to represent the set $\{1,2,\dots,k\}$. 
We denote the elements within a sequence $s$ of $k$ non-negative integers by the indices $i \in [k]$: $s = (s_1, \ldots, s_k)$.
For two such sequences, we write $s \leq s'$ if for all $i \in [k]$, $s_i \leq s'_i$.
Similarly, $s < s'$ if $s \leq s'$ and there exists an index $i^* \in [k]$ such that $s_{i^*} < s'_{i^*}$.
We denote a permutation function on the sequences $s$ by $\sigma$.
There are two types of participants in our model: validators and clients.

\smallskip
\noindent 
\textbf{Validators and Clients.} 
Validators take as input transactions from the environment $\mathcal{Z}$ and execute a blockchain protocol (also known as total order broadcast).
Their goal is to ensure that the clients output a single sequence of transactions.
Validators output consensus messages (\eg blocks, votes), and upon query, send these messages to the clients.
After receiving consensus messages from sufficiently many validators, each client individually outputs a sequence of \emph{finalized} transactions, called the ledger and denoted by $L$.
Clients can be thought of as external observers of the protocol, which can go online or offline at will.

\smallskip
\noindent 
\textbf{Blocks and Chains.} Transactions are batched into blocks and the blockchain protocol orders these blocks instead of ordering transactions individually.
Each block $B_k$ at height $k$ has the following structure: $B_k=(x_k,h_{k}),$
where $x_k$ is the transaction data and $h_{k}=H(B_{k-1})$ is a parent pointer (\eg, a hash) to another block. 
There is a genesis block $B_0 = (\bot, \bot)$ that is common knowledge.
We say that $B$ extends $B'$, denoted by $B' \preceq B$, if $B'$ is the same as $B$, or can be reached from $B$ by following the parent pointers.
Each block that extends the genesis block defines a unique chain.
Two blocks $B$ and $B'$ (or the chains they define) are \emph{consistent} if either $B \preceq B'$ or $B' \preceq B$.
Consistency is a transitive relation.

A client $\client$ \emph{finalizes} a block $B$ at some time $t$ if it outputs $B$ and the chain of blocks preceding $B$ as its ledger at time $t$, \ie, if $\client$'s latest chain contains $B$ for the first time at time $t$. 
The ledger in $\client$'s view is determined by the order of the transactions in this chain.

A blockchain protocol is said to \emph{proceed in epochs of fixed duration} if whenever the protocol is live, a new block is confirmed in the view of any client at a rate of at most one block every $T$ seconds for some constant $T$, \ie, the protocol has bounded chain growth rate.
Such examples include Tendermint \cite{tendermint} and Streamlet \cite{streamlet}, where a new block is proposed by an epoch leader every $2\Delta$ time, where $\Delta$ is a protocol parameter.
These protocols enable the clients to track time by inspecting the timestamps on the blocks.
PBFT-style protocols such as PBFT~\cite{pbft} and HotStuff~\cite{yin2018hotstuff} can also be made to proceed in epochs of fixed duration (despite not being so) by artificially introducing delays before the proposal are broadcast.

\smallskip
\noindent
\textbf{Adversary.} We consider a computationally-bounded adversary $\mathcal{A}$ that can corrupt a fraction of the validators called \emph{adversarial}. 
The remaining ones that follow the protocol are called \emph{honest}.
Adversary controls message delivery subject to the network delay.

\smallskip
\noindent
\textbf{Networking.} 
In a \emph{partially synchronous} network~\cite{DLS88}, the adversary can delay messages arbitrarily until a global stablizaton time (GST) chosen by the adversary. 
After GST, the network becomes synchronous and the adversary must deliver messages sent by an honest validators to its recipients within $\Delta$ time, where $\Delta>0$ is a known delay bound\footnote{We assume synchronized clocks as bounded clock offset can be captured by the delay $\Delta$, and clocks can be synchronized using the process in~\cite{DLS88}.}.
The network is called \emph{synchronous} if GST is known and equal to zero.

\smallskip
\noindent 
\textbf{Security.} Let $L_t^\client$ denote the ledger output by a client $\client$ at time $t$.
We say that a protocol is \emph{safe} if for any times $t,t'$ and clients $\client,\client'$, $L_t^\client$ and $L_{t'}^{\client'}$ are consistent, and 
for any client $\client$, $L_t^\client \preceq L_{t'}^{\client}$ for all $t'\geq t$. 
We say that a protocol is \emph{live} if there is a time $\Tconf>0$ such that for any transaction $\tx$ input to all honest validator at some time $t$, it holds that $\tx\in L^\client_{t'}$ for any client $\client$ and times $t'\geq \max(\GST,t)+\Tconf$.
Note that a protocol satisfying liveness also ensures that clients keep outputting valid transactions; because clients refusing to output invalid transactions as part of their ledgers will not output \emph{anything} after the first invalid transaction.
When we talk about the ledger of a specific protocol $\PI_A$ output by a client $\client$ at time $t$, we will use the notation $L^{\client}_{A, t}$. 

\smallskip
\noindent
\textbf{Certificates.}
We adopt the definition of certificates from~\cite{lewispyeroughgardenccs}.
\begin{definition}[Definition 3.2 of~\cite{lewispyeroughgardenccs}]
We say that a blockchain protocol with confirmation rule $C(.)$ \emph{generates certificates} if the following holds with probability $>1-\epsilon$ when the protocol is run with security parameter $\epsilon$, under the conditions for which safety is satisfied:
There do not exist conflicting ledgers $L_1$ and $L_2$, a time $t$ and sets of consensus messages $M_1$ and $M_2$ broadcast by time $t$, such that $L_i$ is a prefix of the confirmed ledger determined by $C(.)$ on $M_i$, \ie, $C(M_i)$ for $i \in \{1,2\}$.
\end{definition}
An example of a safety condition is over $2/3$ of the validators being honest (\eg, for PBFT~\cite{pbft}), whereas a confirmation rule example, applied to the consensus messages, is confirming a block if there are commit messages for it from over $2/3$ of the validators.
In a certificate-generating protocol, any client $\client$ that finalizes a ledger $L$ can convince any other client to finalize $L$ by showing a subset of the consensus messages.
These messages form a \emph{certificate} for $L$.

All protocols that are safe under partial synchrony generate certificates~\cite{lewispyeroughgardenccs}.
For example, in PBFT-style protocols~\cite{pbft}, Tendermint \cite{tendermint}, HotStuff \cite{yin2018hotstuff} and Streamlet \cite{streamlet}, clients finalize a block upon observing a quorum of commit messages from over $2/3$ of the validator set.
This quorum of commit messages on the block acts as a certificate for the block.
When these protocols are safe, there cannot be two quorums, \ie, certificates, attesting to the finality of conflicting blocks.
In contrast, Nakamoto consensus \cite{bitcoin} does not generate certificates.
As clients confirm a chain \emph{only if} they do not receive a longer chain, no set of messages by themselves suffice to convince clients of the confirmation of a blockchain, as there might always exist a longer but hidden chain of blocks.

\smallskip
\noindent
\textbf{Interchain Consensus Protocols.}
An interchain consensus protocol (interchain protocol for short) is a blockchain protocol, called the \emph{overlay} protocol, executed on top of existing blockchain protocols, called the underlay chains.
Its participants are the clients and validators of the constituent underlay chains $\PI_i$, $i \in [n]$.
All clients and validators observe all underlay chains, but each validator is responsible for participating in the execution of one of these blockchain protocols, which ensures scalability.
Clients and honest validators of each underlay chain run a client of every other chain, and can read from and write to the output ledgers of the other chains.
This restricted communication is captured by the notion of \emph{cross-chain communication (CCC)}~\cite{ccc,trustboost}: each underlay chain $\PI_i$, $i \in [n]$, only exposes read and write functionalities to its finalized ledger.
Clients and validators of every other chain, $\PI_j$, $j \neq i$, verify the finality of $\PI_i$'s ledger via certificates (\eg, by verifying a quorum of signatures on the finalized blocks in PBFT-style protocols~\cite{pbft}), whereas the internal mechanisms and the validator set of $\PI_i$ remain hidden from the interchain protocol, except as used by the CCC to validate certificates.
They write to $\PI_i$ by broadcasting their transactions to all $\PI_i$ validators as input in the presence of a public-key infrastructure, or by using trustless relays that can produce a proof of transmission by collecting replies from sufficiently many $\PI_i$ validators.
Clients of the interchain protocols use only their views of the finalized ledgers of the underlay chains to determine the overlay blockchain's ledger.

The CCC functionality can be implemented by a trusted controller that relays data across chains, or by committees subsampled from among the validators.
A prominent CCC example is the Inter-Blockchain Communication protocol (IBC) of Cosmos~\cite{ibc}, where the messages are transmitted by relayers~\cite{ibc-relayer} akin to controllers.
However, IBC does not require the relayers to be trusted for safety, as it allows the receiver chain's validators to verify messages by inspecting if they were included in the finalized sender chain blocks.

\section{Protocol Primitives}
\label{sec:protocol-primitives}

We build the overlay protocols by composing simpler protocols in two different ways: serial composition and \lvl composition.
In this section, we describe these compositions and their implications for security.

\subsection{Serial Composition}
\label{sec:serial}

\begin{algorithm}[h]
    \captionsetup{font=small} 
    \caption{The algorithm used by a bootstrapping client $\client$ to output the ledger $L_s$ of the serial composition $\PI_s$ instantiated with two constituent blockchains $\PI_A$ and $\PI_B$ at some time $t$. 
    The algorithm takes $L^\client_{B, t}$, the finalized $\PI_B$ ledger output by $\client$ at time $t$, as its input and outputs the ledger $L_s$. 
    The function $\textsc{GetSnapshots}$ returns the snapshots of the $\PI_A$ ledger included in $L^\client_{B, t}$ along with their certificates. 
    The function $\textsc{isCertified}$ returns true if the input ledger is accompanied by a valid certificate.}
    \label{alg.serial}
    \begin{algorithmic}[1]\small
    \Function{\sc OutputChain}{$L^\client_{B, t}$}
        \Let{\snp_1, \ldots, \snp_m}{\textsc{GetSnapshots}(L^\client_{B, t})}
        \Let{L_s}{\bot}
        \For{$i=1,\ldots,m$}
            \If{$\textsc{isCertified}(\snp_i)$}
                \label{line:sanitization}
                \Let{L_s}{\textsc{Clean}(L_s,\snp_i)}
            \EndIf
        \EndFor
        \State\Return $L_s$ 
    \EndFunction
    \end{algorithmic}
\end{algorithm}

\begin{figure}[!htb]
    \hspace*{-2cm}
    \includegraphics[width=1.3\linewidth]{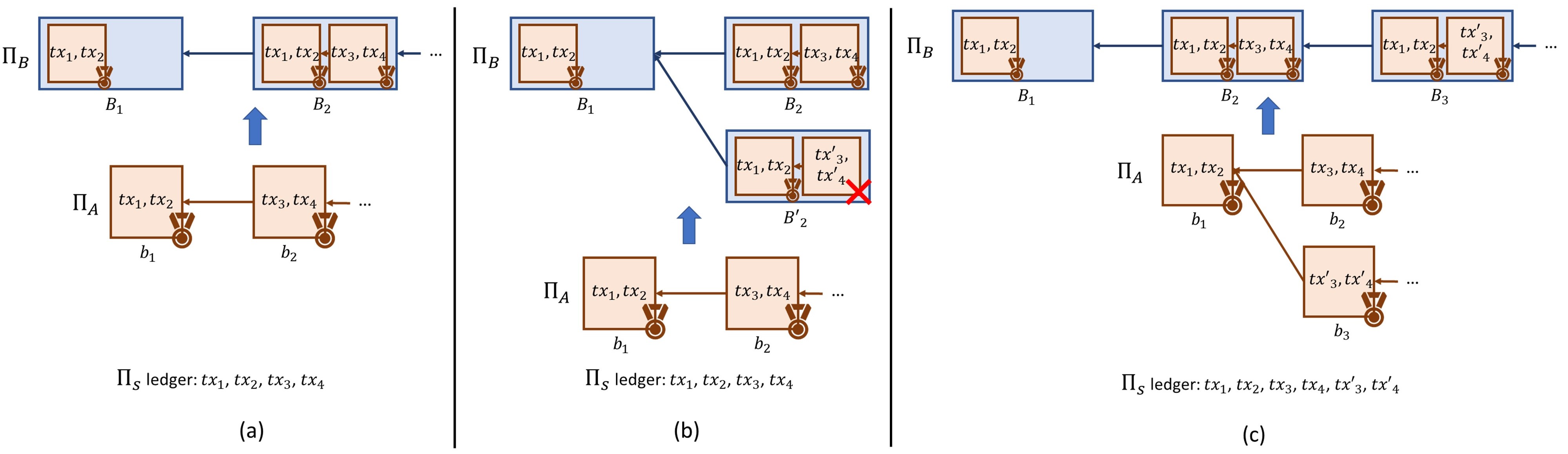}
    \caption{Serial composition. The $\PI_A$ blocks (brown) are denoted by $b_1, b_2, \ldots$ and the $\PI_B$ blocks (blue) are denoted by $B_1, B_2, \ldots$. Certificates of the $\PI_A$ blocks are denoted by the medals. \textbf{In (a)}, both $\PI_A$ and $\PI_B$ are safe. Thus, every client observes the same $\PI_B$ ledger with certified snapshots $\snp_1 = (\tx_1, \tx_2)$ and $\snp_2 = (\tx_1, \tx_2, \tx_3, \tx_4)$. Upon sanitizing the snapshots, clients obtain $\textsc{Clean}(\snp_1, \snp_2) = (\tx_1, \tx_2, \tx_3, \tx_4)$ as the $\PI_s$ ledger. \textbf{In (b)}, the $\PI_B$ ledger is not safe, and two clients $x$ and $y$ observe conflicting $\PI_B$ ledgers $L^{x}_{B,t_1} $ and $L^{y}_{B,t_2}$ with blocks $B_1, B_2$ and $B_1, B'_2$ respectively. The blocks $B_1$, $B_2$ and $B'_2$ contain the certified snapshots $\snp_1 = (\tx_1, \tx_2)$, $\snp_2 = (\tx_1, \tx_2, \tx_3, \tx_4)$ and $\snp'_2 = (\tx_1, \tx_2)$ respectively. Note that $(\tx'_3, \tx'_4)$ is not part of the certified snapshot $\snp'_2$ as they are not included in a certified $\PI_A$ block. Upon sanitizing the snapshots, clients again obtain consistent $\PI_s$ ledgers $L^{x}_{B,t_1} = \textsc{Clean}(\snp_1, \snp_2) = (\tx_1, \tx_2, \tx_3, \tx_4)$ and $L^{y}_{B,t_2} = \textsc{Clean}(\snp_1, \snp'_2) = (\tx_1, \tx_2)$. \textbf{In (c)}, the $\PI_A$ ledger is not safe, and two clients $x$ and $y$ observe conflicting $\PI_A$ ledgers $L^{x}_{A,t_1} $ and $L^{y}_{A,t_2}$ with blocks $b_1, b_2$ and $b_1, b_3$ respectively. However, both clients observe the same $\PI_B$ ledger with blocks $B_1, B_2$, $B_3$ and their certified snapshots $\snp_1, \snp_2$, $\snp_3$. Hence, upon sanitizing the snapshots, clients obtain the same (consistent) $\PI_s$ ledgers $L^{x}_{s,t_1} = L^{y}_{s, t_2} = \textsc{Clean}(\snp_1, \snp_2, \snp_3) = (\tx_1, \tx_2, \tx_3, \tx_4, \tx'_3, \tx'_4)$.}
\label{fig:serial}
\end{figure}

We describe the safety-favoring serial composition $\PI_{s}$ with two constituent certificate-generating blockchain protocols, $\PI_A$ and $\PI_B$ (Fig.~\ref{fig:serial}; \cf Alg.~\ref{alg.serial}).
The $\PI_A$ validators receive transactions from the environment and other validators, and the clients of $\PI_A$ output a certified $\PI_A$ ledger.
Each $\PI_B$ validator acts as a client of $\PI_A$, and consider the $\PI_A$ ledger in its view, called a \emph{snapshot}, and its certificate, as a \emph{transaction input} to $\PI_B$\footnote{For instance, if $\PI_B$ is a blockchain protocol, the snapshots and their certificates will be included in the blocks by the block proposers (\cf Section~\ref{sec:efficiency} for more efficient implementations).} (Fig.~\ref{fig:serial}a).
At any time step $t$, each client $\client$ of the serial composition (which is a client of both $\PI_A$ and $\PI_B$), online at time $t$, inspects the certified snapshots of the $\PI_A$ ledger within its $\PI_B$ ledger.
Then, $\client$ reads the certified $\PI_A$ snapshots in the order they appear in its $\PI_B$ ledger, copies these snapshots and finally eliminates the duplicate transactions appearing in multiple snapshots by calling a sanitization function.
The sanitization function $\clean(L_A,L_B)$ takes two ledgers $L_A$ and $L_B$, concatenates them, eliminates the duplicate transactions that appear in $L_B$ and keeps their first occurrence in $L_A$ (cf. \cite{ebbandflow} \cite{tas2023interchain}). %
Finally, the client outputs the remaining transactions as \emph{its} $\PI_s$ ledger (its view of the $\PI_s$ ledger at that time).
The serial composition satisfies the following security properties:
\begin{theorem}
\label{thm:serial-security}
Consider the serial composition $\PI_s$ instantiated with the certificate-generating blockchain protocols $\PI_A$ and $\PI_B$.
Then, under partial synchrony,
\begin{enumerate}
    \item $\PI_s$ satisfies safety if at least \emph{one} of $\PI_A$ or $\PI_B$ is safe.
    \item $\PI_s$ satisfies liveness with constant latency after GST if \emph{both} $\PI_A$ and $\PI_B$ are live with constant latency after GST.
    \item $\PI_s$ generates certificates.
    \item $\PI_s$ proceeds in epochs of fixed duration if $\PI_A$ and $\PI_B$ proceed in epochs of fixed duration.
\end{enumerate}
\end{theorem}

Proof of Theorem~\ref{thm:serial-security} is given in
Appendix~\ref{sec:proof-serial}.
Proof of the statements 1 and 2 are illustrated by Fig.~\ref{fig:serial} that covers the cases when $\PI_B$ and $\PI_A$ are not safe, yet $\PI_s$ is safe.
Statements 3 and 4 are needed for further composability of the serial composition with other serial and \lvl compositions (\cf the conditions on Theorems~\ref{thm:serial-security} and~\ref{thm:parallel-combined-security}). 
Appendix~\ref{sec:no-certificate} 
describes an attack against the serial composition when $\PI_A$ is not certificate-generating.

For the serial composition $\PI_s$, we require liveness only for the transactions input to all honest $\PI_A$ validators.
In general, liveness must be guaranteed only for the transactions input to \emph{all} honest validators of the underlay protocols.
If validators have access to a public-key infrastructure that identifies each other, then any transaction input to a single honest validator of an underlay protocol can be broadcast to all validators of all underlay protocols, and thus can be included in the ledgers.

\subsection{\Lvl Composition}
\label{sec:view}

A natural liveness-favoring analogue of the serial composition of two blockchains would be a composition that ensures liveness if either of the two chains is live, and safety if both chains are safe.
However, no interchain protocol can satisfy these guarantees, even under synchrony.
Below, we provide the intuition behind this result (\cf Theorems~\ref{lem:converse-general-sync} and~\ref{lem:converse-symmetric-sync} for details).

Consider two blockchains $\PI_A$ and $\PI_B$ that are not live, but safe.
Here, safety of a protocol (\eg, $\PI_A$) means that different clients' views of the $\PI_A$ ledger are consistent, yet it is possible that the $\PI_A$ ledger output by a client conflicts with a $\PI_B$ ledger output by another client.
The protocol $\PI_A$ emulates the behavior of a live blockchain towards a client $\client_1$, whereas it is stalled in $\client_2$'s view, \ie, $L^{\client_2}_{A,t} = \emptyset$ for all times $t$.
In the meanwhile, $\PI_B$ emulates the behavior of a live blockchain towards a different client $\client_2$, whereas it is stalled in $\client_1$'s view, \ie, $L^{\client_1}_{B, t} = \emptyset$ for all times $t$.
Since the \lvl composition is conjectured to be live when either of the blockchains is live, both $\client_1$ and $\client_2$ output transactions based on their observations of the $\PI_A$ and $\PI_B$ ledgers respectively (as far as $\client_1$ is concerned, $\PI_A$ \emph{looks} live, and as far as $\client_2$ is concerned, $\PI_B$ \emph{looks} live).
However, when these $\PI_A$ and $\PI_B$ ledgers are different and conflicting, this implies a safety violation even though both $\PI_A$ and $\PI_B$ are safe, \ie, $\client_1$ and $\client_2$'s $\PI_A$ ($\PI_B$) ledgers are consistent (as $\emptyset$ is a prefix of every ledger).

Given the example above which shows the impossibility of a composition that is live if either chain is live and safe if both are safe, we relax the properties expected of a liveness-favoring composition in two ways:
(i) the \emph{\lvl} composition of $3$ blockchains ensures liveness if $2$ of the $3$ constituent chains are live, and safety if all chains are safe, under partial synchrony (Theorem~\ref{thm:parallel-combined-security}), whereas (ii) the \emph{\lvs} composition of $2$ blockchains, \emph{under synchrony}, ensures liveness if either of the constituent chains is live, and safety if \emph{both} chains are safe and live (Section~\ref{sec:lvs-composition}, Theorem~\ref{thm:lvs-security}).
Here, we focus on the \lvl composition.

For inspiration towards a minimal \lvl composition with these guarantees, we consider a setting, where the protocol participants are validators rather than blockchains.
We observe that a natural analogue of a blockchain that is not live, but safe, is a validator with \emph{omission faults}.
Since the \lvl composition for blockchains requires the safety of all constituent protocols for safety, its analogue for validators would tolerate only omission faults.
Thus, our \lvl composition is motivated by omission fault tolerant (OFT) consensus protocols \cite{lamport2019part,budhiraja1993primary,oki1988viewstamped}.  
Before presenting the composition, we briefly describe these OFT protocols for validators, which we extend to the blockchain setting.

\subsubsection{The OFT Protocol for Validators.}
\label{sec:oft-protocol}
The OFT protocol is a leader-based blockchain protocol that generates certificates under a partially synchronous network.
It is run by $3$ validators mirroring the most basic \lvl composition. %
It proceeds in epochs of fixed duration $3\Delta$. 
In a nutshell, it works as follows:

Each epoch $v$ has a unique leader that proposes a block at the beginning of the epoch %
, \ie, at time $3\Delta v$.
Upon observing a proposal for epoch $v$, validators broadcast acknowledge messages for the proposed block at time $3\Delta v + \Delta$. 
Upon observing a \emph{certificate} of $2$ unique acknowledge messages from epoch $v$ for the epoch's proposal, validators and clients finalize the proposed block and its prefix chain.
If a validator does not observe a certificate of $2$ acknowledge messages for an epoch $v$ proposal by time $3\Delta v + 2\Delta$, it broadcasts a leader-down message for epoch $v$, where the message contains the block with the highest epoch number among the ones it previously voted for.
Leader-down messages enable the leader of the next epoch to identify the correct block to propose on to preserve safety.
A detailed protocol description is presented in
Appendix~\ref{sec:oft-protocol-detailed}.

\subsubsection{From OFT Protocol to the \Lvl Composition.}
\label{sec:oft-to-parallel-const}

\begin{figure}[!htb]
\centering
    \centering
    \includegraphics[width=\linewidth]{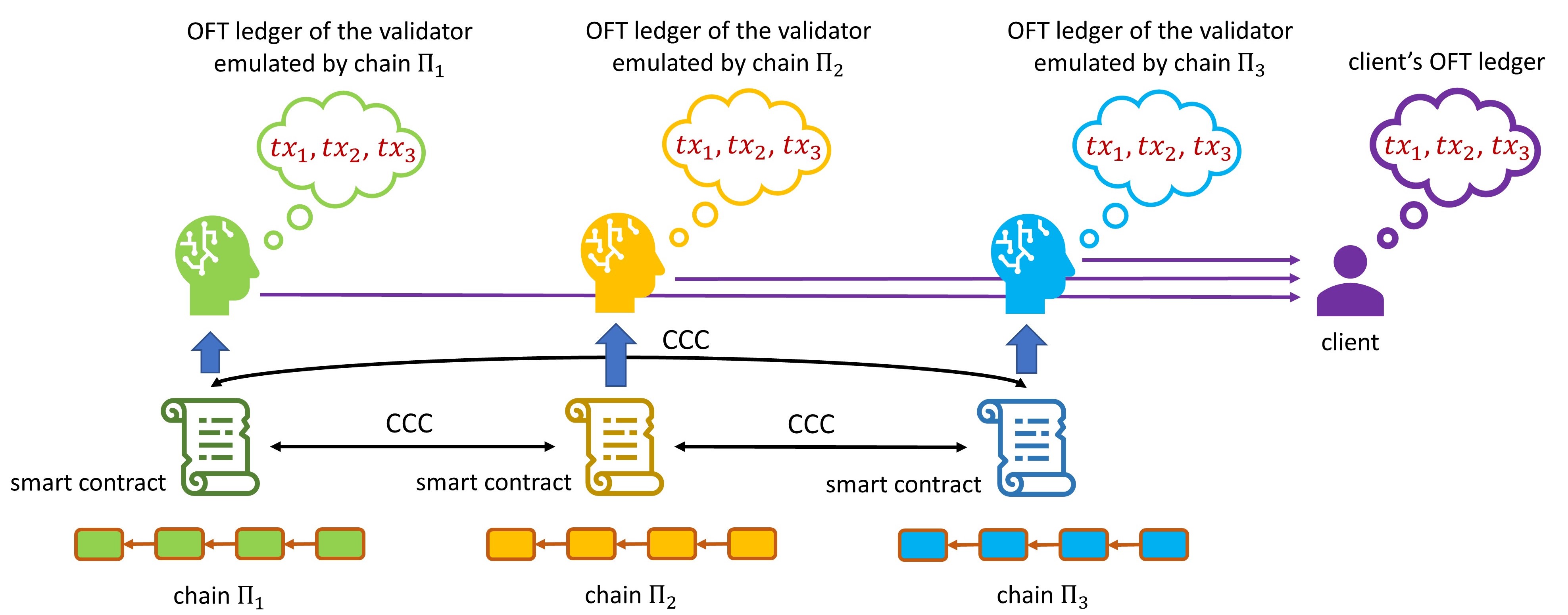}
    \caption{\Lvl composition. An overlay OFT protocol run on top of $3$ underlay blockchains. A smart contract on each of the underlays emulates a validator of the OFT protocol and outputs a finalized OFT ledger. The client reads the underlay chains' ledgers and outputs the OFT ledger finalized by a majority of the emulated validators.}
    \label{fig:lvs}
\end{figure}

We next describe a \lvl composition for $3$ blockchains.
It consists of $3$ underlay blockchain protocols, $\PI_A$, $\PI_B$ and $\PI_C$, run by validators and an overlay protocol, \ie, the OFT protocol, run on top of these chains (Fig.~\ref{fig:lvs}).
Each underlay protocol executes a smart contract that \emph{emulates a validator} of the overlay OFT protocol (\cf
Appendix~\ref{sec:connect} 
for a discussion on validator emulation).
These emulated validators exchange messages via the CCC abstraction.
There is a PKI that identifies on each underlay chain the $2$ other chains emulating a validator (\eg, by means of the public keys of the other chains' validators).

\smallskip
\noindent
\emph{Blockchains.}
The \lvl composition requires the underlay protocols to run general-purpose smart contracts and to proceed in epochs of fixed duration $T$.
This is because the overlay OFT protocol requires each emulated validator to keep track of the time passed since it entered any given epoch.
In general, it is impossible to emulate the validators of any overlay protocol secure under partial synchrony, if the underlay protocol has no means of keeping track of the real time (\cf
Appendix~\ref{sec:async-impossibility}).
We achieve this functionality by using underlay protocols that proceed in epochs of fixed time duration such as Tendermint~\cite{tendermint} or Streamlet~\cite{streamlet} (\cf
Appendix~\ref{sec:async-impossibility}).
However, our \lvl composition can also be instantiated with optimistically responsive protocols (\cf Section~\ref{sec:conclusion} for more discussion).

Using epoch numbers recorded in the underlay blocks, the smart contract tracks the time passed since it entered any given epoch of the overlay protocol.
If it entered some epoch $v$ of the overlay protocol at time $t$, it moves to epoch $v+1$ at an underlay block of an underlay epoch $3\Tconf / T$, where $\Tconf$ is the cross-chain communication latency.
Here, the $3\Delta$ epoch of the overlay OFT protocol is replaced by a $3\Tconf$ length epoch, since the messages exchanged by the \emph{emulated} validators incur additional latency, including the finalization latency of the underlay chains besides network delay.

\smallskip
\noindent
\emph{Clients.}
We next describe how clients of the \lvl composition output a ledger for the overlay protocol using the ledgers of the underlay chains.
Upon outputting a ledger $L_A$, $L_B$ and $L_C$ for each underlay protocol $\PI_A$, $\PI_B$ and $\PI_C$, at some time $t$, a client inspects the execution of the smart contracts as attested by these ledgers.
If the execution trace on some ledger is invalid according to the rules of the smart contract, then the client discards the parts of the ledger starting with the first invalid transaction recorded on it, thus turning invalid execution into a liveness failure.
For instance, sending a syntactically incorrect message is detectable by only inspecting the messages on a ledger.
In contrast, sending two acknowledge messages for conflicting overlay blocks in the same overlay epoch might not be detected upon inspection, since these two messages can exist in separate execution traces emulating the same OFT validator, \ie, on conflicting ledgers, observed by different clients (safety failure).

Once the client observes the execution traces for the validators emulated on valid portions of the ledgers $L_A$, $L_B$ and $L_C$, it identifies the blocks of the overlay protocol committed by each emulated validator.
It accepts and outputs an overlay block and its prefix chain if it was committed by $2$ or more emulated validators (as attested by the ledgers of $2$ or more underlay chains).
If a client accepts and outputs an overlay block and its chain of height $h$, it never outputs a shorter overlay chain from that point on.
If the client observes multiple conflicting $L_A$, $L_B$ and $L_C$ ledgers when the safety of underlay chains is violated, it considers all these ledgers, and might output conflicting overlay blocks as a result.
However, this is not a problem, as the proof of the next theorem shows that the client will continue to output blocks and retain liveness nevertheless.

The \lvl composition satisfies the following:

\begin{theorem}
\label{thm:parallel-combined-security}
Consider the \lvl composition $\PI_t$ instantiated with the protocols $\PI_A$, $\PI_B$ and $\PI_C$, that proceed in epochs of fixed duration.
Then, under partial synchrony,
\begin{enumerate}
    \item $\PI_t$ satisfies safety if all of $\PI_A$, $\PI_B$ and $\PI_C$ are safe.
    \item $\PI_t$ satisfies liveness with constant latency after GST if $2$ blockchains among $\PI_A$, $\PI_B$ and $\PI_C$ are live after GST with constant latency and proceed in epochs of fixed duration.
    \item $\PI_t$ generates certificates if $\PI_A$, $\PI_B$ and $\PI_C$ do so. 
    \item $\PI_t$ proceeds in epochs of fixed duration.
\end{enumerate}
\end{theorem}

Proof of Theorem~\ref{thm:parallel-combined-security} is given in
Appendix~\ref{sec:proof-parallel}.
Statements 1 and 2 are based on the proof of the original OFT protocol design for validators.
Statements 3 and 4 are needed for further composability of the \lvl composition with other serial and \lvl composition (\cf the conditions on Theorem~\ref{thm:serial-security}).

\section{Circuits for Partial Synchrony}
\label{sec:partial-synchrony}

In this section, we construct overlay protocols via circuit composition achieving the security properties claimed by Theorem~\ref{thm:informal} and show optimality by proving Theorem~\ref{thm:impos}.
We also extend these results to all possible overlay protocols, akin to the generalization of security properties to quorum and fail-prone systems.
Unlike the security claims for the protocol primitives, all of the proofs below are algebraic in nature.

\subsection{Extended Serial and \Lvl Constructions}
\label{sec:extended-lvl}

We first build extended serial and \lvl constructions as building block toward the full circuit composition.

\begin{lemma}
\label{lem:extended-serial}
Let $\PI_i$, $i \in [k]$ be $k$ different blockchain protocols that generate certificates.
Then, there exists a protocol, called the $n$-serial composition, satisfying the following properties:
\begin{itemize}
\item it is safe if at least one of $\PI_i$, $i \in [k]$ is safe.
\item it is live after $\GST$ with constant latency if all of $\PI_i$, $i \in [k]$ are live after GST with constant latency.
\item it generates certificates.
\item it proceeds in epochs of fixed duration if all of $\PI_i$, $i \in [k]$ do so.
\end{itemize}
\end{lemma}
Lemma~\ref{lem:extended-serial} follows directly from iteratively applying Theorem~\ref{thm:serial-security} on the protocols $\PI_i$, $i \in [k]$.

\begin{lemma}
\label{lem:root}
For any integer $f\geq 1$, let $\PI_i$, $i \in [2f+1],$ be $2f+1$ different blockchain protocols that proceed in epochs of fixed duration.
Then, there exists a protocol, called the $(2f+1)$-\lvl composition, satisfying the following properties:
\begin{itemize}
\item it is safe if all $\PI_i$, $i \in [2f+1]$ are safe.
\item it is live after $\GST$ with some constant latency if at least $f+1$ of $\PI_i$, $i \in [2f+1]$ are live after GST with some constant latency.
\item it generates certificates if all $\PI_i$, $i \in [2f+1]$ generate certificates.
\item it proceeds in epochs of fixed duration.
\end{itemize}
\end{lemma}

Note that the original \lvl composition with three protocols (referred to as the \lvl composition) would be called a $3$-\lvl composition.
Proof of Lemma~\ref{lem:root} is presented in
Appendix~\ref{sec:parallel-3-3-2}.
It constructs an $f$-\lvl composition via strong induction on the number $f$.
The inductive step uses the $(f-1)$-\lvl composition and the serial composition of Lemma~\ref{lem:extended-serial}, whereas the base case follows from the properties of the $3$-\lvl composition shown by Theorem~\ref{thm:parallel-combined-security}.

\subsection{Permutation Invariant Circuits for Partial Synchrony}
\label{sec:circuit-symmetric-psync}

We next prove Theorem~\ref{thm:informal}, which characterizes the security of so-called permutation invariant overlay protocols.
This class of protocols achieves safety (or liveness) as long as \emph{any} subset of the underlay chains with a sufficient size provide the same set of security guarantees (\eg, any subset of $5$ out of $7$ underlay chains is all safe and/or all live).
In this sense, these protocols do not distinguish between the underlay chains.

Proof of Theorem~\ref{thm:informal} relies on Theorems~\ref{thm:serial-security} and~\ref{thm:parallel-combined-security}.
Recall that a tuple $(k,s,l)$ was defined to be achievable if there exists an interchain protocol with $k$ blockchains such that
if at least $s$ blockchains are safe, the protocol is safe, and
if at least $l$ blockchains are live, the protocol is live.
By definition, a serial composition achieves the $(2,1,2)$ point (Theorem~\ref{thm:serial-security}), and a \lvl composition achieves the $(3,3,2)$ point (Theorem~\ref{thm:parallel-combined-security}).
Similarly, a $(2f+1)$-\lvl composition achieves the $(2f+1, 2f+1, f+1)$ point, and there exist such compositions for any $f \geq 1$ by Lemma~\ref{lem:root}, which itself follows from Theorems~\ref{thm:serial-security} and~\ref{thm:parallel-combined-security}.

For two blockchain protocols $\Pi_A, \Pi_B$, we denote by $\Pi_A\oplus \Pi_B$ the serial compositions of these two blockchains as described in Section~\ref{sec:serial}. 
Consider $k$ protocols $\Pi_1,\Pi_2,\dots,\Pi_k$. 
We iteratively define $\oplus_{i=1}^{j+1}\Pi_i=\pp{\oplus_{i=1}^{j}\Pi_i}\oplus \Pi_{j+1}$ for $j \in [k-1]$,
and denote a protocol achieving the $(k,1,k)$ point by 
$
\pi^{(k,1,k)}(\PI_1,\dots,\PI_{k}).
$
and denote a protocol achieving the $(2f+1,2f+1,f+1)$ point by 
$
\pi^{(2f+1,2f+1,f+1)}(\PI_1,\dots,\PI_{2f+1}).
$

Towards the final result,
the following lemma shows that we can construct a protocol achieving $(k+m,s,l+m)$ point using a protocol achieving $(k,s,l)$.
\begin{lemma}\label{lem:leave}
For any integer $m\geq 1$, %
if $(k,s,l)$ is achievable using $\pi^{(k,s,l)}$
then, $(k+m,s,l+m)$ is achievable using
\begin{equation}\label{pi:n_small}
\begin{aligned}
    &\pi^{(k+m,s,l+m)}(\{\PI_i\}_{i=1}^{k+m}) =&\mathop{\oplus}_{\substack{S\subseteq [k+m]\\|S|=k}}\pi^{(k,s,l)}\pp{\{\PI_{j}\}_{j\in S}}.
\end{aligned}
\end{equation}
\end{lemma}

\begin{proof}
We first show that $\pi^{(k+m,s,l+m)}$ defined in \eqref{pi:n_small} is safe if at least $s$ of the blockchains are safe. There exists a subset $S_0\subseteq [k+m]$ with $|S_0|=k$ such that at least $s$ blockchains among $\{\Pi_j\}_{j\in S_0}$ are safe. This implies that $\pi^{(k,s,l)}\pp{\{\PI_{j}\}_{j\in S_0}}$ is safe. 
As we enumerate all subsets with size $k$ in constructing $\pi^{(k,s,l)}\pp{\{\PI_{j}\}_{j\in S_0}}$, by Lemma~\ref{lem:extended-serial}, we observe that $\pi^{(k+m,s,l+m)}(\{\PI_i\}_{i=1}^{k+m})$ is safe; as one of the blockchains in the serial composition, namely $\pi^{(k,s,l)}\pp{\{\PI_{j}\}_{j\in S_0}}$, is safe.

On the other hand, suppose that at least $l+m$ of the blockchains are live, which implies that at most $k-l$ blockchains are not live. Therefore, for any arbitrary choice of size-$k$ subset $\{\PI_{j}\}_{j\in S}$ from $\{\PI_j\}_{j=1}^{k+m}$, at most $k-l$ blockchains are not live, or equivalently, at least $l$ blockchains in $\{\PI_{j}\}_{j\in S}$ are live. This implies that $\pi^{(k,s,l)}\pp{\{\PI_{j}\}_{j\in S}}$ is live for all possible choices of subset $S$ with size $k$. 
Then, by Lemma~\ref{lem:extended-serial}, we observe that $\pi^{(k+m,s,l+m)}(\{\PI_i\}_{i=1}^{k+m})$ is live; as all of the blockchains in the serial composition are live.
\end{proof}

Finally, we present the proof of Theorem \ref{thm:informal}.

\begin{proof}[Proof of Theorem~\ref{thm:informal}]
By Lemma~\ref{lem:root}, there are circuit compositions achieving the $(2f+1,2f+1,f+1)$ point given copies of \emph{any} two protocols achieving the $(2,1,2)$ and $(3,3,2)$ points respectively, via recursive compositions of these protocols.
This in turn implies that for any given integers $k,s,l$ such that $\flr{k/2}+1\leq l\leq n$ and $s=2(k-l)+1$ (boundary of the achievable points on Fig.~\ref{fig:achieve}), the point
$(s,s,k-l+1)=(2(k-l)+1,2(k-l)+1,k-l+1)$
is achievable.
Since $s=2(k-l)+1$ for these boundary points, we have $l+s-k=k-l+1$. 
Therefore, by Lemma~\ref{lem:leave},
$(k,s,l)=(s+(k-s),s,k-l+1+(k-s))$
is achievable.
This implies that all points $(k,s,l)$ such that $\flr{k/2}+1\leq l\leq n$ and $s \geq 2(k-l)+1$ are achievable.
\end{proof}

\subsection{General Circuits for Partial Synchrony}
\label{sec:circuit-general-psync}

We next present a general characterization of the security of the overlay protocol under partial synchrony as a function of the safety and liveness of the underlay chains.
Note that a general characterization would include protocols that are not permutation invariant, \ie, providing different security guarantees when two subsets of the underlay chains achieve the same set of security properties.
For example, a non-permutation invariant overlay protocol with three underlay chains $\PI_i$, $i \in [3]$, might be live when the underlay chains $\PI_1$ and $\PI_2$ are both live, but it might not be so when $\PI_1$ and $\PI_3$ are live.
For such protocols, our notation of $(k,s,l)$ tuples fall short of characterizing the security properties.
Therefore, we develop a new model for the security of overlay protocols synthesized from $k$ underlay chains.

\subsubsection{The Model}
\label{sec:model-1}

As the security of a blockchain consists of safety and liveness, we use $s,l\in\{0,1\}^k$ to denote the list of predicates indicating which underlay chains are guaranteed to be safe and live.
Specifically, $s_i=1$ if the $i$-th underlay chain is guaranteed to be safe, and $s_i=0$ if the $i$-th chain is not guaranteed to be safe.
Then, the security properties of an overlay protocol $\PI$ can be characterized by two sets $V^S, V^L\subseteq 2^{\{0,1\}^{2k}}$,
which express the dependence of $\PI$'s safety and liveness on the safety and liveness of the underlay chains.
Namely, $(s,l)\in V^S$ if the overlay protocol $\PI$ is guaranteed to be safe when for all $i$ such that $s_i=1$, the $i$-th underlay chain is guaranteed safety, and for all $j$ such that $l_j=1$, the $j$-th chain is guaranteed liveness.
Similarly, $(s,l)\in V^L$ if the overlay protocol $\PI$ is guaranteed to be live when for all $i$ such that $s_i=1$, the $i$-th underlay chain is guaranteed safety, and for all $j$ such that $l_j=1$, the $j$-th chain is guaranteed liveness.
We hereafter use the $(V^S,V^L)$ characterization of security in lieu of the $(k,s,l)$ tuples.
Given these definitions, any set $V^S$, $V^L$ for an overlay protocol satisfies
\begin{itemize}
\item[(P1)] If $v\in V$, $w\geq v$, then $w\in V$. 
\end{itemize}
A sequence $v \in V$ is called an \emph{extreme} element if there is no $w \in V$ such that $w < v$.
Let $\text{exm}(V)$ be the set containing all extreme elements in $V$.
By property (P1), $\text{exm}(V)$ uniquely describes $V$, and any protocol $\PI$ can be characterized by the two sets $E^S, E^L \subseteq 2^{\{0,1\}^{2k}}$ consisting of the extreme elements in $V^S$ and $V^L$: 
$E^S = \text{exm}(V^S)$ and $E^L = \text{exm}(V^L)$.

\subsubsection{The Result}

Given the model above, the security properties achievable by overlay protocols under partial synchrony can be described as follows.
For $s\in\{0,1\}^n$, let us define $\ind(s)=\{i:s_i=1\}$. 

\begin{theorem}
\label{lem:circuit-general-psync}
For any tuple $(E^S,E^L) \subseteq 2^{\{0,1\}^{2k}}$ such that
\begin{enumerate}
    \item For all $(s,l)\in E^L$, $(s',l') \in E^S$, it holds that $s = l' = 0^k$.
    \item For all $(0^k,l^1), (0^k,l^2)\in E^L, (s,0^k)\in E^S$, it holds that $\ind(l^1) \cap \ind(l^2)\cap \ind(s) \neq \emptyset$.
\end{enumerate}
there exists an overlay protocol characterized by a tuple dominating $(E^S,E^L)$.
\end{theorem}

The proof is in
Appendix~\ref{sec:app-circuit-general-psync} 
and inductively constructs the desired overlay protocol.
Intuitively, Theorem~\ref{lem:circuit-general-psync} states that safety (liveness) of the overlay protocol depends only on the safety (liveness) of the underlay chains, and any two quorums of underlay chains required for the liveness of the overlay protocol must intersect at a chain whose safety is required for the safety of the overlay protocol.
Note that Theorem~\ref{lem:circuit-general-psync} implies Theorem~\ref{thm:informal}, as Theorem~\ref{lem:circuit-general-psync} characterizes security for all types of overlay protocols, including permutation-invariant ones analyzed by Theorem~\ref{thm:informal}.
We opted to present Theorem~\ref{thm:informal} first for ease of understanding.

\section{The Converse for Partial Synchrony}
\label{sec:converse-psync}

\subsection{The Converse Result for Partial Synchrony}
\label{sec:converse-general-psync}

We start with the converse result that applies to all overlay protocols under partial synchrony, showing the optimality of our security characterization in Theorem~\ref{lem:circuit-general-psync}.
\begin{theorem}
\label{lem:converse-general-psync}
Let $\PI$ be an overlay blockchain protocol. 
Let $(s^i,l^i)\in\{0,1\}^n$ for $i\in[3]$ satisfy 
$(s^1,l^1)\in V^L, (s^2,l^2)\in V^L, (s^3,l^3)\in V^S.$
Then, we have $\ind(l^1)\cap\ind(l^2)\cap\ind(s^3)\neq\emptyset$. 
\end{theorem}

Note that the converse exactly matches the second clause of Theorem~\ref{lem:circuit-general-psync}.

\begin{proof}
For contradiction, suppose $\ind(l^1)\cap\ind(l^2)\cap\ind(s^3) = \emptyset$.
Denote the $k$ underlay blockchains by $\PI_1,\dots,\PI_k$. 
There are two clients $\client_1,\client_2$.
Consider the following three worlds. 

\noindent \textbf{World 1:} 
All blockchains are safe. 
The underlay chains $\PI_i$, $i\in\ind(l^1)$ are live, and the others are stalled. 
The adversary sets $\GST=0$. 
Suppose that $\tx_1$ is input to the protocol at time $t=0$. 
As $(s^1,l^1)\in V^L$, the overlay blockchain is live. 
At time $t_1=\Tconf$, the client $\client_1$ outputs $\tx_1$ as its interchain ledger: $L_{t_1}^{\client_1}=[\tx_1]$.

\noindent \textbf{World 2:} 
All blockchains are safe. 
The underlay chains $\PI_i$, $i\in\ind(l^2)$ are live, and the others are stalled. 
The adversary sets $\GST=0$. 
Suppose that $\tx_2$ is input to the protocol at time $t=0$. 
As $(s^2,l^2)\in V^L$, the overlay blockchain is live. 
At time $t_2=\Tconf$, the client $\client_2$ outputs $\tx_2$ as its interchain ledger: $L_{t_2}^{\client_2}=[\tx_2]$.

\noindent\textbf{World 3:} All blockchains are live. 
The underlay chains $\PI_i$, $i\in \ind(l^1)\cap\ind(l^2)$ are not safe, and the others, including those in $s^3$, are safe. 
For simplicity, let $Q=\ind(l^1)\cap\ind(l^2)$.
The adversary sets $\GST=2 \Tconf$ and creates a network partition before GST such that client $\client_1$ can only communicate with the validators in $\PI_{i}$, $i\in \ind(l^1)$, and client $\client_2$ can only communicate with the validators in $\PI_{i}$, $i\in \ind(l^2)$.
As a result, for client $\client_1$, the underlay chains $\PI_{i}$, $i \notin \ind(l^1)$ seem stalled until at least time $2\Tconf$, and for client $\client_2$, the underlay chains $\PI_{i}$, $i \notin \ind(l^2)$ seem stalled until at least time $2\Tconf$, and 
Suppose that $\tx_1,\tx_2$ are input to the protocol at time $t=0$.
However, the adversary reveals $\tx_1$ only to the honest validators in $\PI_i$ for $i\in \ind(l^1)/Q$ and $\tx_2$ only to those in $\PI_i$ for $i\in \ind(l^2)/Q$. 
Moreover, it delays any communication between the validators in $\PI_i$ for $i\in \ind(l^1)/Q$ and those in $i\in \ind(l^2)/Q$ until after GST.

As the chains $\PI_i$, $i\in Q$ are not safe, they can simultaneously interact with $\client_1$ and the chains $\PI_i$, $i\in \ind(l^1)/Q$ as in World 1 and with $\client_2$ and the chains $\PI_i$, $i\in \ind(l^2)/Q$ as in World 2.
To ensure this, the adversary delays any messages from the honest validators of the chains $\PI_i$, $i\in Q$, to $\client_1$ and $\client_2$, except the certificates received by $\client_1$ and $\client_2$ in Worlds 1 and 2 respectively.
As we assume $\PI_i$, $i\in Q$, are not safe, such certificates attesting to conflicting ledgers must exist for the chains $\PI_i$, $i\in Q$.
Then, client $\client_1$ cannot distinguish World 1 and World 3 before $2\Tconf$, which implies that $L^{\client_1}_{t_1} = [\tx_1]$.
Similarly, client $\client_2$ cannot distinguish World 2 and World 3, which implies that $L^{\client_2}_{t_2} = [\tx_2]$. 
However, $L^{\client_1}_{t_1}$ and $L^{\client_2}_{t_2}$ conflict with each other, which violates the safety of the overlay protocol.
This is a contradiction; as all underlay chains are live, and those in $s^3$ are safe.
\end{proof}

\subsection{The Converse Result for Permutation Invariant Protocols under Partial Synchrony}
\label{sec:converse-symmetric-psync}

Before proving Theorem~\ref{thm:impos}, we introduce a more comprehensive notation for permutation invariant overlay protocols to capture the fact that the safety (or liveness) of the overlay protocol can depend on \emph{both} the safety and liveness of the underlay chains.
Although Theorem~\ref{thm:impos} shows that the liveness of the underlay chains do not help achieve better safety properties for the overlay (and vice versa), we nevertheless need a notation that allows the \emph{possibility} of such cross-dependence between safety and liveness to argue for the absence of this cross-dependence.
Moreover, we will use the new notation for permutation invariant overlay protocols later to describe the achievable security guarantees under synchrony, where the safety of the overlay protocol \emph{depend} on the liveness of the underlay chains.

\subsubsection{The Model for Permutation Invariant Overlay Protocols}
\label{sec:model-2}
Permutation invariance means that the overlay blockchain treats the underlays identically.
\begin{definition}[Permutation Invariance]
We say a protocol $\PI$ is permutation invariant if the sets $V^S$ and $V^L$ both satisfy that
\begin{itemize}
\item[(P2)] If $v\in V$, then $\sigma(v)\in V$ for all permutations $\sigma$.
\end{itemize}
Here, $v=(s,l)$, and we define $\sigma(v)=(\sigma(s),\sigma(l))$, where $\sigma(s),\sigma(l) \in\{0,1\}^k$, $\sigma(s)_i=s_{\sigma(i)},\ \sigma(l)_i = l_{\sigma(i)}\ \forall i\in[k]$. 
\end{definition}

By property (P2), we can create an equivalence relation `$\sim$' over the sets in $V$.
We say $v \sim w$, if there exists a permutation $\sigma:[k]\to[k]$ such that $w=\sigma(v)$.
The relation `$\sim$' defines a quotient set $V/\sim$, which is the set of equivalence classes in $V$.
Given $v=(s,l)$ and
$$
c_s(v):=\#\{i:s_i=1\},\;c_l(v):=\#\{i:l_i=1\}, c_{sl}(v):=\#\{i:s_i=l_i=1\},
$$
each equivalence class $\{\sigma(v)|\sigma:[k]\to[k] \text{ is a permutation}\}$ is uniquely represented by a tuple $(c_s(v),c_l(v),c_{sl}(v))\in\mbN^3$. 
As the set $\text{exm}(V)$ (for either $V^S$ or $V^L$) containing all extreme elements also satisfies (P2), 
we can also partition $\text{exm}(V)/\sim$ into equivalence classes, each represented by a tuple $(n_s,n_l,n_{sl})\in\mbN^3$. 
Therefore, given property (P1), the set $V$ can be represented by a set of tuples
$
P=\{(n_s^{(i)},n_l^{(i)},n_{sl}^{(i)})|i\in \mbN\}.
$

Finally, any permutation invariant protocol $\PI$ can be characterized by two sets $P^S,P^L\in2^{\mbN^3}$, representing $V^S$ and $V^L$ respectively
and interpreted as follows:
For any $(n_s,n_l,n_{sl})\in P^S$ and $(m_s,m_l,m_{sl})\in P^L$, we have
\begin{itemize}
\item $\PI$ is safe if at least $n_s$ blockchains are safe, $n_l$ blockchains are live, and $n_{sl}$ blockchains are safe and live.
\item $\PI$ is live if at least $m_s$ blockchains are safe, $m_l$ blockchains are live, and $m_{sl}$ blockchains are safe and live.
\end{itemize}
Let $V(P):=\{v|c_s(v)\geq n_s,c_l(v)\geq n_l,c_{sl}(v)\geq n_{sl}, (n_s,n_l,n_{sl})\in P\}$.

For two set pairs $(P^S,P^L)$ and $(\tilde P^S, \tilde P^L)$ characterizing permutation invariant overlay protocols, we define the partial order
$(P^S,P^L) \succeq (\tilde P^S, \tilde P^L)$ to mean that $V(P^S) \supseteq V(\tilde P^S)$ and $V(P^L)\supseteq V(\tilde P^L)$.
For $v\in\{0,1\}^k$, let us define $\ind(v)=\{i:s_i=1\}$.
\begin{lemma}
\label{lem:partial_order}
$(P^S,P^L)\succeq (\tilde P^S, \tilde P^L)$ if and only if for any $\tilde p_1\in \tilde P^S,\tilde p_2\in \tilde P^L$,
there exists $p_1\in P^S$ and $p_2\in P^L$ such that $p_1\leq \tilde p_1$ and $p_2\leq \tilde p_2$.
\end{lemma}
\begin{proof}
It is sufficient to show that $V(P)\supseteq V(\tilde P)$ if and only if for any $\tilde p\in \tilde P$, there exists $p\in P$ such that $p\leq \tilde p$.
Suppose that we have $V(P)\supseteq V(\tilde P)$. 
For any $\tilde p=(\tilde n_s,\tilde n_l,\tilde n_{sl})\in \tilde P$, consider $\tilde v\in\{0,1\}^{2k}$
such that $c_s(\tilde v)=\tilde n_s,c_l(v)=\tilde n_l,c_{sl}(v)=\tilde n_{sl}$. Then, $\tilde v\in V(P)$. 
There exists an extreme element $v\in V$ such that $v\leq \tilde v$. 
Defining $p=(c_s(v),c_l(v),c_{sl}(v))$, we can conclude that $p\leq \tilde p$. 
Suppose that for any $\tilde p\in \tilde P$, there exists $p\in P$ such that $p\leq \tilde p$. For any $\tilde v\in V(\tilde P)$,
there exists $\tilde p=(\tilde n_s,\tilde n_l,\tilde n_{sl})\in \tilde P$ such that $c_s(\tilde v)=\tilde n_s,c_l(\tilde v)=\tilde n_l,c_{sl}(\tilde v)=\tilde n_{sl}$.
Let $p\in P$ such that $p\leq \tilde p$.
From the definition of $V(P)$, we have $\tilde v\in V(P)$.
\end{proof}

\subsubsection{The Result}

The converse result for permutation invariant overlay protocols under partial synchrony, Theorem~\ref{thm:impos}, follows as a corollary of Theorem~\ref{lem:converse-general-psync}.
It shows the optimality of our security characterization in Theorem~\ref{thm:informal}.
Its proof is presented in
Appendix~\ref{sec:app-converse-symmetric-psync}.

\begin{theorem}[Theorem~\ref{thm:impos}]
\label{lem:converse-symmetric-psync}
Let $\PI$ be a permutation invariant overlay blockchain protocol characterized by $(P^S, P^L)$.
Consider the tuples $(m_s, m_l, m_{sl}) \in P^L, (n_s, n_l, n_{sl}) \in P^S.$
Then, it holds that $n_s \geq 2(k-m_l)+1$ and $m_l > k/2$.
\end{theorem}

\section{Circuits for Synchrony}
\label{sec:synchrony}

In this section, we construct overlay protocols via circuit composition achieving the security properties claimed by Theorem~\ref{thm:informal-sync}, and show their optimality.
As the properties achievable under synchrony are stronger than those achievable under partial synchrony, to bridge the gap between partial synchrony and synchrony, we first introduce the \lvs composition as a new protocol primitive in addition to the serial and \lvl compositions (\cf Section~\ref{sec:view} for a discussion of the \lvl and \lvs compositions).
We then state the security result for general overlay protocols using the model in Section~\ref{sec:model-1}.
Equipped with the model in Section~\ref{sec:model-2}, we subsequently show the security properties claimed for permutation invariant overlay protocols by Theorem~\ref{thm:informal-sync} as a corollary of the general result.
We end with a proof of optimality for both results.

\subsection{\Lvs Composition}
\label{sec:lvs-composition}

\begin{algorithm}[ht]
    \captionsetup{font=small} 
    \caption{The algorithm used by a client $\client$, online for at least $\Tconf$ time, to output the ledger $L^\client_{p,t}$ of the \lvs composition $\PI_p$ instantiated with two constituent blockchains $\PI_A$ and $\PI_B$ at some time $t$. 
    The algorithm takes the ledgers $L^\client_{A,t-\Tconf}$, $L^\client_{A,t}$, $L^\client_{B,t-\Tconf}$ and $L^\client_{B,t}$ output by $\client$ at time $t$ and outputs the ledger $L^\client_{p,t}$. 
    The function $\textsc{Interleave}(L_1,L_2)$ with inputs of same length returns the interleaved ledger $L_o$ such that $L_o[2i-1] = L_1[i]$ and $L_o[2i] = L_2[i]$ for all $i \in [|L_1|]$.}
    \label{alg.lvs}
    \begin{algorithmic}[1]\small
    \Function{\sc OutputChain}{$L^\client_{A,t-\Tconf}, L^\client_{A,t}, L^\client_{B,t-\Tconf},  L^\client_{B,t}$}
    \Let{\ell}{\min(|L^\client_{A,t-\Tconf}|,|L^\client_{B,t-\Tconf}|)}
    \If{$L^\client_{A,t-\Tconf} \subseteq L^\client_{B,t} \land L^\client_{B,t-\Tconf} \subseteq L^\client_{A,t}$}
        \Let{L^\client_{p,t}}{\textsc{Interleave}(L^\client_{A,t}[{:}\ell],L^\client_{B,t}[{:}\ell])}
    \Else
        \Let{i^*}{\mathsf{argmax}_{i \in \{A,B\}}|L^\client_{i,t-\Tconf}|}
        \Let{L^\client_{p,t}}{\textsc{Interleave}(L^\client_{A,t}[{:}\ell],L^\client_{B,t}[{:}\ell]) || L^\client_{i^*,t}[\ell{:}]}
    \EndIf
    \State\Return $L^\client_{p,t}$ 
    \EndFunction
    \end{algorithmic}
\end{algorithm}

We now describe the \lvs composition with two underlay chains, $\PI_A$ and $\PI_B$ (Alg.~\ref{alg.lvs}).
Upon getting a transaction from the environment, every honest $\PI_A$ and $\PI_B$ validator broadcasts the transaction to every other validator.

Let $L^\client_{A,t}$, $L^\client_{B,t}$ and $L^\client_{p,t}$ denote the $\PI_A$, $\PI_B$ ledgers and the ledger of the \lvs overlay protocol in the view of a client $\client$ at time $t$.
Consider a client $\client$ that has been online for at least $\Tconf$ time.
It obtains the \lvs ledger as a function of the $\PI_A$ and $\PI_B$ ledgers.
For this purpose, $\client$ first checks if every transaction in $L^\client_{A,t-\Tconf}$ appears in $L^\client_{B,t}$, and if every transaction in $L^\client_{B,t-\Tconf}$ appears within $L^\client_{A,t}$ (\emph{the interleaving condition}).
If so, it \emph{interleaves} the prefixes of the $\PI_A$ and $\PI_B$ ledgers to construct the $\PI_p$ ledger:
\begin{equation}
\label{eq:lvs-1}
L^\client_{p,t} := \textsc{Interleave}(L^\client_{A,t}[{:}\ell],L^\client_{B,t}[{:}\ell]),
\end{equation}
where $\ell := \min(|L^\client_{A,t-\Tconf}|,|L^\client_{B,t-\Tconf}|)$.
$\textsc{Interleave}$ function applied on equal size ledgers $L_1$, $L_2$ outputs a ledger $L_o$ such that $L_o[2i-1] = L_1[i]$ and $L_o[2i] = L_2[i]$ for all $i \in [|L_1|]$.

If the interleaving condition fails, then $\client$ interleaves the prefixes of the two ledgers and outputs the remainder of the longer ledger: defining $i^* = \mathsf{argmax}_{i \in \{A,B\}}|L^\client_{i,t-\Tconf}|$, it sets
\begin{equation}
\label{eq:lvs-2}
L^\client_{p,t} := \textsc{Interleave}(L^\client_{A,t}[{:}\ell],L^\client_{B,t}[{:}\ell]) || L^\client_{i^*,t}[\ell{:}].
\end{equation}

The \lvs composition satisfies the security properties below:
\begin{theorem}
\label{thm:lvs-security}
Consider the $\lvs$ composition $\PI_p$ instantiated with the blockchain protocols $\PI_A$ and $\PI_B$.
Then, under synchrony,
\begin{enumerate}
    \item $\PI_p$ satisfies safety if \emph{both} $\PI_A$ and $\PI_B$ are safe and live.
    \item $\PI_p$ satisfies liveness with constant latency if \emph{either} $\PI_A$ or $\PI_B$ is live with constant latency.
    \item $\PI_p$ generates certificates if both $\PI_A$ and $\PI_B$ do so.
    \item $\PI_t$ proceeds in epochs of fixed duration if $\PI_A$ and $\PI_B$ do so.
\end{enumerate}
\end{theorem}

Proof of Theorem~\ref{thm:lvs-security} is presented in
Appendix~\ref{sec:lvs-security}.
It shows that if both chains are safe and live, the interleaving condition is satisfied, ensuring the safety of the $\PI_p$ ledger.
If either chain is live, then all transactions up to the length of the longer chain is output, ensuring the liveness of the $\PI_p$ ledger.
Note that the parallel composition does not satisfy accountable safety despite satisfying safety. 
This is not too surprising since its safety requires both the safety and liveness of the underlay chains.

\subsection{General Circuits for Synchrony}
\label{sec:circuit-general-sync}

\begin{theorem}
\label{lem:circuit-general-sync}
For any tuple $(E^S,E^L) \subseteq 2^{\{0,1\}^{2k}}$ such that
\begin{enumerate}
    \item For all $(s,l)\in E^L$ it holds that $s = 0^k$,
    \item For all $(0^k,l^1), (0^k,l^2)\in E^L, (s,l)\in E^S$, it holds that 
    \begin{enumerate}
        \item either there are indices $i \in \ind(l^1)$ and $j \in \ind(l^2)$ such that $(s_i, l_i, s_j, l_j) = (1,1,1,1)$,
        \item or $\ind(l^1) \cap \ind(l^2)\cap \ind(s) \neq \emptyset$, 
    \end{enumerate}
\end{enumerate}
there exists an overlay protocol characterized by a tuple dominating $(E^S,E^L)$.
\end{theorem}

We present the proof of Theorem~\ref{lem:circuit-general-sync} in
Appendix~\ref{sec:appendix-circuit-general-sync}.
It shows achievability by constructing a circuit very similar to that constructed by Theorem~\ref{lem:circuit-general-psync}.
Intuitively, Theorem~\ref{lem:circuit-general-sync} states that for the safety of the overlay protocol, either any two quorums of underlay chains required for the liveness of the overlay protocol must both contain at least one safe and live chain (which can be different), or these quorums must intersect at a safe chain.

\subsection{Permutation Invariant Circuits for Synchrony}
\label{sec:circuit-symmetric-sync}

Following lemma proves the achievability guarantees claimed by Theorem~\ref{thm:informal-sync}.
It uses the notation of Section~\ref{sec:model-2} and follows as a corollary of Theorem~\ref{lem:circuit-general-sync}. 

\begin{theorem}[Theorem~\ref{thm:informal-sync}, Achievability]
\label{lem:circuit-symmetric-sync}
If a protocol is characterized by $(P^S,P^L)$ within
$$
\{(\{(2(k-m_l)+1,0,0),(0,0,k-m_l+1)\},\{(0,m_l,0)\})|k/2<m_l\leq k\},
$$
then there exists a permutation invariant overlay protocol characterized by a tuple dominating $(P^S,P^L)$ under synchrony.
\end{theorem}

\section{The Converse for Synchrony}
\label{sec:converse-sync}

We now show the optimality of our security characterization in Theorem~\ref{thm:informal-sync}.

\subsection{The Converse Result for Synchrony}
\label{sec:converse-general-sync}

We start with the converse result that applies to all overlay protocols under synchrony, showing the optimality of our security characterization in Theorem~\ref{lem:circuit-general-sync}.

\begin{theorem}
\label{lem:converse-general-sync}
Let $\PI$ be an overlay blockchain protocol. 
Let $(s^i, l^i) \in \{0,1\}^n$ for $i \in [3]$ satisfy 
$(s^1,l^1)\in V^L, (s^2,l^2)\in V^L, (s^3,l^3)\in V^S.$
Then, it holds that 
\begin{enumerate}
    \item either there are indices $i \in \ind(l^1)$ and $j \in \ind(l^2)$ such that $(s_i, l_i, s_j, l_j) = (1,1,1,1)$,
    \item or $\ind(l^1) \cap \ind(l^2)\cap \ind(s) \neq \emptyset$.
\end{enumerate}
\end{theorem}

Proof of Theorem~\ref{lem:converse-general-sync} is presented in
Appendix~\ref{sec:app-converse-general-sync}, 
and relies on an indistinguishability argument between different worlds like the proof of Theorem~\ref{lem:converse-general-psync}.
Note that the converse exactly matches the clause of Theorem~\ref{lem:circuit-general-sync}.

Theorem~\ref{lem:converse-general-sync} is reduced to Theorem~\ref{lem:converse-general-psync} 
if the set of security functions are restricted to those, where safety of the overlay protocol depends only on the safety of the underlay chains, and the liveness of the overlay protocol depends only on the liveness of the underlay chains.
This is consistent with \cite[Theorem B.1]{forensics}, which proves that the accountable safety-liveness tradeoff under synchrony is the same as the safety-liveness tradeoff under partial synchrony.

\subsection{The Converse Result for Permutation Invariant Protocols under Synchrony}
\label{sec:converse-symmetric-sync}

The converse result for permutation invariant overlay protocols under synchrony as claimed in Theorem~\ref{thm:informal-sync}, follows as a corollary of Theorem~\ref{lem:converse-general-sync}.
It shows the optimality of our security characterization in Theorem~\ref{thm:informal-sync}.

\begin{theorem}[Theorem~\ref{thm:informal-sync}, Converse]
\label{lem:converse-symmetric-sync}
Let $\PI$ be a permutation invariant overlay blockchain protocol characterized by $(P^S, P^L)$.
Consider the tuples 
$(m_s, m_l, m_{sl}) \in P^L, (n_s, n_l, n_{sl}) \in P^S.$
Then, it holds that either $n_{sl} \geq k-m_l+1$ or $n_s \geq 2(k-m_l)+1$ and $m_l > k/2$.
\end{theorem}

Theorem~\ref{lem:converse-symmetric-sync} follows as a corollary of Theorem~\ref{lem:converse-general-sync}, and its proof is in
Appendix~\ref{sec:app-converse-symmetric-sync}.

Appendix~\ref{sec:appendix-pareto}
summarizes all of the results in Sections~\ref{sec:partial-synchrony},~\ref{sec:converse-psync},~\ref{sec:synchrony} and~\ref{sec:converse-sync} by identifying all pareto-optimal protocols under partial synchrony and synchrony using the language developed in Sections~\ref{sec:model-1} and~\ref{sec:model-2}.

\section{Efficiency}
\label{sec:efficiency}

Our model of interchain consensus protocols in Section~\ref{sec:prelim} allows the validators of the underlay blockchains to read the ledgers of the other underlays.
For each validator, this implies a communication load proportional to the number of underlay chains, \ie, low scalability.
However, all of our compositions (serial, \lvl and \lvs) can be modified to retain their security properties when the validators merely run light clients of the other chains.
For instance, the serial composition in Section~\ref{sec:serial} can be instantiated with \emph{succinct timestamps} as in~\cite{tas2023interchain}; so that the constituent protocols receive and order timestamps of the blocks of the other protocols rather than snapshots of the whole ledger.
Here, timestamps can even be made less frequent for more efficiency (albeit at the expense of latency).
When the timestamps are implemented with binding hash functions, their ordering suffices to resolve forks on the other chains and ensure safety as long as any chain is safe.

The \lvl construction in Section~\ref{sec:view} requires the validators of each underlay chain to only follow a smart contract, dedicated to executing the OFT protocol, on the other chains to react to their OFT protocol messages.
This again warrants at most a light client functionality.
Finally, in the \lvs construction of Section~\ref{sec:lvs-composition}, the validators of the underlay blockchains need not communicate at all except broadcasting the circuit-composition related transactions.
In turn, external observers (clients) are responsible for interleaving their ledgers.
Therefore, all three composition methods, and by induction, our circuit constructions can work with underlay validators running light clients of each other's chains.
Implementation of these constructions with light clients is left as future work.

\section{Conclusion}
\label{sec:conclusion}
In this work, we have analyzed the security of interchain consensus protocols under synchrony and partial synchrony.
We next outline a few open questions and future directions implied by our work.
As our serial composition requires the underlay chains to produce certificates and the protocols secure under the sleepy network model~\cite{sleepy} (dynamic availability~\cite{kiayias2017ouroboros}, unsized setting~\cite{lewispyeroughgardenccs}) do not generate certificates~\cite{lewispyeroughgardenccs}, our results do not extend to underlay chains secure under the (synchronous) sleepy network model (no protocol can be secure under both partial synchrony and the sleepy network model~\cite{blockchain-cap}).
It is thus an open question to design a serial composition for underlay chains that do not generate certificates.

Although we have instantiated the \lvl composition with underlay chains that proceed in fixed time durations, the composition can also work with \emph{optimistically responsive} protocols.
These protocols achieve latency that is $O(\delta)$, where $\delta$ is the real-time network delay, under optimistic conditions.
They can keep track of time with the help of an oracle committee of so-called \emph{time keepers}~\cite{trustboost} that input the real time into the protocol.
Another alternative that does not require any trust in oracles is for the smart contracts on the underlay chains to adaptively estimate time.
For instance, if the contracts notice that the overlay protocol has not made progress while the underlay protocols have, it can slow down the underlay protocols.
It is future work to formalize the details of these solutions.

Our recursive compositions of circuits could require an underlay blockchain to appear in exponentially many sub-circuits.
Our goal in this work was to show the achievability of the properties proven for the interchain consensus protocols.
For small numbers of underlay chains, our results coupled with the optimizations in Section~\ref{sec:efficiency} still yield practical constructions for the safety-favoring points.
It is an open question to design more scalable interchain consensus protocols for all points.

\nocite{full-version}

\bibliography{references}

\appendix
\section{Blockchains vs. Validators}
\label{sec:connect}

In this section, we investigate an analogue of our interchain circuit construction for validators as protocol participants rather than blockchains.
This connection was first made by Recursive Tendermint \cite{recursive-tendermint}, which proposed implementing a consensus protocol, originally developed for validators, on top of Tendermint-based blockchains representing those validators and using IBC communication for message passing.
As our construction was developed for blockchains that produce certificates, our validators can be assumed to have authenticated communication, \ie, all messages are accompanied with unforgeable signatures.
As argued in Section~\ref{sec:view}, a natural analogue of a chain that is safe, but not live is a validator with omission faults that can drop incoming and outgoing messages arbitrarily, but cannot sign conflicting messages.
In this context, what should be the fault model for a validator corresponding to a blockchain that is not safe, but live?
For such a validator, we take inspiration from the rollback attacks on trusted execution environments (TEEs), where the TEE equivocates by rolling its state back to a previous point and creating two conflicting execution traces from that point onward \cite{rollback1,rollback2,rollback3}.
In the same vein, a validator corresponding to a not safe but live chain emulates the behavior of multiple fictitious validators simultaneously.
Each fictitious validator creates an execution trace that is computationally indistinguishable from the trace of an honest validator, yet might be conflicting with the traces of other fictitious validators.
We call a validator with this fault model a \emph{split-brain} validator.

As in Section~\ref{sec:view}, when the consensus protocol has an execution environment (\eg, the Ethereum Virtual Machine), we use omission faults to also capture the behavior of a validator that outputs a \emph{single} invalid execution trace.
This is because an invalid trace is detectable and can be rejected by other validators, essentially reducing this validator to an omission-fault validator.

A natural conclusion of the above model is that a blockchain that is neither safe nor live should represent a validator that is simultaneously a split-brain and omission-fault validator.
As such, it emulates multiple fictitious validators with conflicting executions, each of which, taken independently, is an omission-fault validator.
Note that this is a strictly weaker model of protocol violation than a Byzantine adversary that can deviate from the protocol arbitrarily.
A funny, but all the same, real example of a Byzantine validator is a computer that runs the protocol code honestly along with a software that draws diagrams of blockchains.
Although such validators are not captured by the notion of a split-brain and omission-fault validator, we argue that for almost all consensus protocols, a Byzantine validator cannot do any action that is not accessible to a split-brain omission-fault validator and has an effect on the security of the consensus protocol.
Therefore, we refer to a split-brain, omission-fault validator as a Byzantine validator.

Finally, given our mapping from blockchains to validators, we can reinterpret our results in Sections~\ref{sec:protocol-primitives} and~\ref{sec:partial-synchrony} in the context of Figure~\ref{fig:achieve}.
The x-axis of the figure, which used to denote the fraction of non-live blockchains now identifies the total fraction of Byzantine or omission-fault validators.
Similarly, the y-axis of the figure, which used to denote the fraction of non-safe blockchains now identifies the total fraction of Byzantine or split-brain validators.
Based on the results in Sections~\ref{sec:protocol-primitives},~\ref{sec:partial-synchrony}, we conjecture that any point marked blue is achievable by a consensus protocol, whereas no consensus protocol can achieve the red points.

\section{Attack on Serial Composition}
\label{sec:no-certificate}

Here, we present an attack on the serial composition of non certificate-generating protocols.
For the proof of Theorem~\ref{thm:serial-security}, it is crucial that the protocol $\PI_A$ generates certificates.
To illustrate this, suppose $\PI_A$ is Nakamoto consensus, which does not generate certificates.
Then, even if there is no safety violation on $\PI_A$, the adversary can create two forks of the $\PI_A$ chain, one shorter, \ie, with less proof-of-work (PoW), than the other.
Let $L_{\mathsf{shrt}}$ and $L_{\mathsf{long}}$ denote the shorter and the longer forks of the $\PI_A$ chain respectively.
When $\PI_B$ loses safety and the clients observe multiple conflicting forks of the $\PI_B$ ledger, the adversary can include snapshots of $L_{\mathsf{shrt}}$ and $L_{\mathsf{long}}$ in different $\PI_B$ forks.
Note that there is no safety violation in $\PI_A$, as upon observing both $L_{\mathsf{shrt}}$ and $L_{\mathsf{long}}$, clients would unanimously output $L_{\mathsf{long}}$ as the finalized $\PI_A$ chain.
However, since $\PI_B$ is no longer safe, clients observing the $\PI_B$ fork containing $L_{\mathsf{short}}$ now output $L_{\mathsf{short}}$ as the $\PI_s$ ledger, and the clients observing the $\PI_B$ fork containing $L_{\mathsf{long}}$ output $L_{\mathsf{long}}$ as the $\PI_s$ ledger.
This implies a safety violation even though at least one of the blockchains, namely $\PI_A$, is still safe, contradicting with Theorem~\ref{thm:serial-security}.

\section{Impossibility of Partially Synchronous Overlay Protocols on Responsive Underlay Chains}
\label{sec:async-impossibility}

It is impossible to emulate the validators of an overlay protocol that is secure under partial synchrony on underlay protocols that do not proceed in epochs of fixed duration (\eg, on responsive \cite{roic} or optimistically responsive \cite{yin2018hotstuff} protocols).
If an overlay protocol is secure under partial synchrony and relies on timeouts, then the underlay protocols must be able to track time via epochs, \ie, should not be responsive.
On the other hand, if an overlay protocol is secure under partial synchrony and does not use timeouts, it must then be secure in an asynchronous network.
By the celebrated FLP result, if this protocol is deterministic, it must have a non-terminating execution that violates its liveness \cite{flp}.
To overcome the implications of the FLP result and ensure liveness under asynchrony, most protocols employ a common coin that generates a sequence of unpredictable coin tosses \cite{commoncoin,roic}.
However, the unpredictability of the future tosses requires validators to keep a secret and reveal it at the time of the coin toss (\eg, when common coins are implemented via threshold signatures, the validators must keep their secret signing keys private and contribute signatures at the time of the toss).
When such validators are emulated on a smart contract running on a public ledger, it is not possible to embed a secret into the contract that cannot be read until the time of the coin toss.

Although it is possible to design underlay blockchains that can keep a secret, such designs require secret sharing among the validators of the underlay blockchains as part of the consensus protocol, \ie, a redesigning of the underlays \cite{secretchain}.
In this case, an overlay protocol would rely on the validators of the underlay blockchains to reveal their secrets so that the emulated validators of the overlay can contribute unpredictable bits to the common coin.
However, this scheme violates the definition of interchain consensus protocol that only uses the blockchain ledgers and do not involve the validators of the underlay blockchains as part of the protocol apart from verifying the finality of the ledgers of the underlay blockchains (\cf Section~\ref{sec:prelim}).
Therefore, in the interchain protocol paradigm, it is impossible to emulate validators of an overlay protocol that is secure under asynchrony via smart contracts on underlay blockchains that expose public ledgers.
This implies that the overlay protocol that is secure under partial synchrony must use timeouts and the underlay protocols must proceed in epochs of fixed duration.

\section{Detailed Description of the OFT Protocol for Validators}
\label{sec:oft-protocol-detailed}

The OFT protocol is a leader-based protocol that proceeds in epochs of $3\Delta$ time.
It uses three types of messages: propose, acknowledge, and leader-down.
Let $n=2f+1$ denote the total number of validators, where $f$ is a protocol parameter.
For the parallel composition in Section~\ref{sec:view}, $f=1$ is sufficient.
A collection of acknowledge messages $\lra{\text{ack}(B,v)}$ for a block $B$ and epoch $v$ from $f+1$ validators is called an epoch $v$ certificate for block $B$ and denoted by $C_v(B)$.
We assume that the genesis block $B_0$ is a certified block at epoch $0$ for all validators.
Each validator keeps track of the proposal it has observed with the highest epoch number.

In each epoch $v$, a leader $L_v$ is selected via round-robin, \ie, the index of the leader $L_v$ is $(v \text{ mod }n)$. 
At the beginning of the epoch, $L_v$ builds a new block depending on the way it entered epoch $v$:
\begin{enumerate}
\item If $L_v$ has previously received acknowledge messages (\ie, $\lra{\text{ack}(B_{k-1},v-1)}$) for a block $B_{k-1}$ and the previous epoch $v-1$ from $f+1$ validators, \ie, an epoch $v-1$ certificate $C_{v-1}(B_{k-1})$ for block $B_{k-1}$, it builds and proposes a block $B_{k}$ with the parent pointer $h_{k}=H(B_{k-1})$. 
\item If $L_v$ has previously received leader-down messages (\ie, $\lra{\text{leader-down}(B, v-1)}$) for the previous epoch $v-1$ from $f+1$ validators, where each message carries a block $B$, it builds and proposes a block $B_{k}$ with the parent pointer $h_k=H(B_{k-1})$, where $B_{k-1}$ is the block with the maximal epoch number among the blocks within the $f+1$ leader-down messages.
\item If $L_v$ has not received $f+1$ acknowledge or leader-down messages by the beginning of epoch $v$, it does not build or propose a new block.
\end{enumerate}
As any leader must have observed either $f+1$ acknowledge (\ie, a certificate) or $f+1$ leader-down messages for epoch $v-1$ to build a block for epoch $v$, we call them \emph{tickets} for epoch $v$.

The detailed protocol description is given below:
\begin{enumerate}
\item \textbf{Enter.} Epoch $v$ starts at time $3 \Delta v$. Upon entering the epoch, $L_v$ builds a block $B_k$ and broadcasts a proposal message $\lra{\text{prop} (B_k, v)}$ as described above, along with the ticket. If $L_v$ has not received $f+1$ acknowledge or leader-down messages by the beginning of the epoch, it does not build or propose a new block.
\item \textbf{Acknowledge.} At time $3 \Delta v + \Delta$, if a validator has observed a proposal message $\lra{\text{prop}(B_k,v)}$ by the epoch leader $L_v$, it broadcasts an acknowledge message $\lra{\text{ack}(B_{k},v)}$ for the proposed block $B_k$ and epoch $v$. 
\item \textbf{Leader-down.} At time $3 \Delta v + 2\Delta$, if a validator has not yet observed $f+1$ acknowledge messages $\lra{\text{ack}(B_{k},v)}$ for a block $B_k$ proposed at epoch $v$ by $L_v$, it sends a leader-down message $\lra{\text{leader-down}(B_{k'}, v)}$ for epoch $v$, where $B_{k'}$ is the block associated with the maximal epoch number $v'$ at which the validator previously sent an acknowledge message $\lra{\text{ack}(B_{k'},v')}$.
\end{enumerate}

All validators \textbf{commit} a block $B$ proposed at some epoch $v$ by the correct leader $L_v$ and its prefix chain upon observing the proposal $\lra{\text{prop}(B,v)}$ and a certificate $C_v(B)$ for block $B$ and epoch $v$.

\section{Pareto-Optimal Overlay Protocols}
\label{sec:appendix-pareto}

We define pareto-optimality and explicitly identify all pareto-optimal protocols under partial synchrony and synchrony using the language developed in Sections~\ref{sec:model-1} and~\ref{sec:model-2} and combining the earlier results.
\begin{definition}[Pareto-Optimality]
    \label{def:pareto-optimality}
    We say that an overlay protocol $\PI$ \emph{dominates} an overlay protocol $\PI'$ (denoted by $\PI \succeq \PI'$)
    if the associated set pairs $(V^S_\PI,V^L_\PI)$ and $(V^S_{\PI'},V^L_{\PI'})$ satisfy 
    $V^S_{\PI'}\subseteq V^S_{\PI},\ V^L_{\PI'}\subseteq V^L_{\PI}.$
    We say that $\PI$ \emph{strictly dominates} $\PI'$ (denoted by $\PI \succ {\PI'}$)
    if the associated set pairs $(V^S_\PI,V^L_\PI)$ and $(V^S_{\PI'},V^L_{\PI'})$ satisfy
    $V^S_{\PI'}\subseteq V^S_{\PI},\ V^L_{\PI'}\subseteq V^L_{\PI},$
    and either $V^S_{\PI'} \subset V^S_{\PI}$ or $V^L_{\PI'} \subset V^L_{\PI}$.
    We say that an overlay blockchain protocol $\PI$ is \emph{pareto-optimal} if 
    there exists no $\PI'$ such that $\PI' \succ \PI$.
\end{definition}

\begin{theorem}[Partial Synchrony]
\label{thm:psync-formal}
Consider the set of tuples $(E^S,E^L) \subseteq 2^{\{0,1\}^{2k}}$ that satisfy the following properties:
\begin{enumerate}
    \item for all $(s,l)\in E^S, (s',l') \in E^L,$ it holds that $s' = l = 0^k$, 
    \item for all $(0^k,l^1), (0^k,l^2)\in E^L, (s,0^k)\in E^S,$ it holds that $\ind(l^1)\cap\ind(l^2)\cap\ind(s)\neq\emptyset$, and
    \item for all $(s, 0^k) \in E^S$, for all $s' < s$, there exist $(0^k,l^1), (0^k,l^2)\in E^L$ such that $\ind(l^1)\cap\ind(l^2)\cap\ind(s') = \emptyset$.
\end{enumerate}
Under a partially synchronous network, all pareto-optimal overlay protocols are characterized by the tuples $(E^S,E^L)$ within this set.
\end{theorem}

Theorem~\ref{thm:psync-formal} follows from Theorems~\ref{lem:circuit-general-psync} and~\ref{lem:converse-general-psync}.
Theorem~\ref{lem:circuit-general-psync} constructs overlay protocols characterized by the $(E^S,E^L)$ tuples in the set described above under partial synchrony.
Then, Theorem~\ref{lem:converse-general-psync} shows that no overlay protocol can achieve the security properties required of any element outside this set under partial synchrony.

\begin{theorem}[Synchrony]
\label{thm:sync-formal}
Consider the set of tuples $(E^S,E^L) \subseteq 2^{\{0,1\}^{2k}}$ that satisfy the following properties:
\begin{enumerate}
    \item For all $(s,l)\in E^L$ it holds that $s = 0^k$.
    \item For all $(0^k,l^1), (0^k,l^2)\in E^L, (s,l)\in E^S$, it holds that 
    \begin{enumerate}
        \item either there are indices $i \in \ind(l^1)$, $j \in \ind(l^2)$ such that $(s_i, l_i, s_j, l_j) = (1,1,1,1)$,
        \item or $\ind(l^1) \cap \ind(l^2)\cap \ind(s) \neq \emptyset$. 
    \end{enumerate}
    \item For all $(s, l) \in E^S$, for all $(s',l') < (s, l)$, there exist $(0^k,l^1),\\ (0^k,l^2) \in E^L$ such that 
    \begin{enumerate}
        \item for any $(i,j)$ such that $i \ind(l^1)$, $j \in \ind(l^2)$, $(s'_i, l'_i, s'_j, l'_j) \neq (1,1,1,1)$,
        \item and $\ind(l^1)\cap\ind(l^2)\cap\ind(s') = \emptyset$. 
    \end{enumerate}
\end{enumerate}
Under a partially synchronous network, all pareto-optimal overlay protocols are uniquely described by $(P^S,P^L)$ within this set.
\end{theorem}

Theorem~\ref{thm:sync-formal} follows from Theorems~\ref{lem:circuit-general-sync} and~\ref{lem:converse-general-sync}.
Theorem~\ref{lem:circuit-general-sync} constructs overlay protocols characterized by the $(E^S,E^L)$ tuples in the set described above under synchrony.
Then, Theorem~\ref{lem:converse-general-sync} shows that no overlay protocol can achieve the security properties required of any element outside this set under partial synchrony.

\begin{theorem}[Permutation Invariance, Partial Synchrony]
\label{thm:perm-invariant-psync-formal}
Under a partially synchronous network, all pareto-optimal permutation invariant overlay protocols are characterized by $(P^S,P^L)$ within:
$$
\{(\{(2(k-m_l)+1,0,0)\},\{(0,m_l,0)\})|k/2<m_l\leq k\}
$$
\end{theorem}

Theorem~\ref{thm:perm-invariant-psync-formal} follows from Theorems~\ref{thm:informal} and~\ref{thm:impos}, \ie, Theorem~\ref{lem:converse-symmetric-psync}.
Theorem~\ref{thm:impos} constructs overlay protocols for all $(P^S,P^L)$ in the set $\{(\{(2(k-m_l)+1,0,0)\},\{(0,m_l,0)\})|k/2<m_l\leq k\}$ under partial synchrony.
Theorem~\ref{thm:impos} then shows that no overlay protocol can achieve the security properties required of any element outside this set under partial synchrony.

\begin{theorem}[Permutation Invariance, Synchrony]
\label{thm:perm-invariant-sync-formal}
Under a synchronous network, all pareto-optimal permutation invariant overlay protocols are characterized by $(P^S,P^L)$ within:
$$
\{(\{(2(k-m_l)+1,0,0),(0,0,k-m_l+1)\},\{(0,m_l,0)\})|k/2<m_l\leq k\}
$$
\end{theorem}

Theorem~\ref{thm:perm-invariant-sync-formal} follows from Theorem~\ref{thm:informal-sync}, \ie, Theorems~\ref{lem:circuit-symmetric-sync} and~\ref{lem:converse-symmetric-sync}.
Theorem~\ref{lem:circuit-symmetric-sync} constructs overlay protocols for all extreme elements in $(P^S,P^L) = \{(\{(2(k-m_l)+1,0,0),(0,0,k-m_l+1)\},\{(0,m_l,0)\})|k/2<m_l\leq k\}$ under synchrony.
Theorem~\ref{lem:converse-symmetric-sync} shows that no overlay protocol can achieve the security properties of the elements outside the set under synchrony.

\section{Security Proofs}
\label{sec:proofs}

\subsection{Proof of Theorem~\ref{thm:serial-security}}
\label{sec:proof-serial}

The proof uses the following propositions:
\begin{proposition}
\label{prop:consistency-2}
Consider two different ledgers $L^x$ and $L^y$.
If $L_x \preceq L_y$, then $L_x \preceq \clean(L^x,L^y)$.
\end{proposition}
Proof follows from the definition of the sanitization function $\clean(.,.)$.

\begin{proposition}
\label{prop:consistency}
Consider three different ledgers $L^x$, $L^y$ and $L^z$.
If the ledgers $L^x$, $L^y$ and $L^z$ are all consistent with each other, then $\clean(L^x,L^y)$ is consistent with $L^z$.
\end{proposition}
\begin{proof}
Without loss of generality, we can assume that $L^x \preceq L^y$. Then, we have $\clean(L^x,L^y)=L^y$. 
As $L^y$ and $L^z$ are consistent, it follows that $\clean(L^x,L^y)$ is consistent with $L^z$. 
\end{proof}

\begin{proof}[Proof of Theorem~\ref{thm:serial-security}]
We first prove the safety claim.
Suppose $\PI_A$ satisfies safety (Fig.~\ref{fig:serial}b).
Consider two clients $x$ and $y$ (which can be the same client) who observe the (potentially conflicting) $\PI_B$ ledgers $L^{x}_{B,t_1}$ and $L^{y}_{B,t_2}$ at times $t_1$ and $t_2$ respectively.
As $\PI_A$ generates certificates, all $\PI_A$ ledgers with certificates of finality are consistent with each other.
Therefore, all certified snapshots $\snp_i$, $i = 1,2, \ldots$, of the $\PI_A$ ledger appearing within $L^{x}_{B,t_1}$ are consistent with each other and the certified snapshots $\snp'_i$, $i = 1,2, \ldots$, of the $\PI_A$ ledger appearing within $L^{y}_{B,t_2}$ (and vice versa).
Note that to obtain the $\PI_s$ ledgers $L^{x}_{s,t_1}$ and $L^{y}_{s,t_2}$, clients $x$ and $y$ iteratively run the sanitization step at Line~\ref{line:sanitization} of Alg.~\ref{alg.serial} on the snapshots $\snp_i$, $i = 1,2, \ldots$, and $\snp'_i$, $i = 1,2, \ldots$, respectively.
Then, by iteratively applying Proposition~\ref{prop:consistency} on these snapshots, we observe that the $\PI_s$ ledger $L^{x}_{s,t_1}$ obtained by client $x$ at time $t_1$ is consistent with the $\PI_s$ ledger $L^{y}_{s,t_2}$ obtained by client $y$ at time $t_2$.
Hence, for any clients $x$ and $y$ and times $t_1$ and $t_2$, the $\PI_s$ ledgers in the view of $x$ and $y$ at times $t_1$ and $t_2$ are consistent.

Now, suppose $\PI_B$ satisfies safety (Fig.~\ref{fig:serial}c).
Then, for any two clients $x$ and $y$ and times $t_1$ and $t_2$, the $\PI_B$ ledgers $L^{x}_{B,t_1}$ and $L^{y}_{B,t_2}$ are consistent.
This implies that the sequence of certified snapshots $(\snp_1, \ldots)$ of the $\PI_A$ ledger appearing within $L^{x}_{B,t_1}$ is a prefix of the sequence of certified snapshots $(\snp'_1, \ldots)$ of the $\PI_A$ ledger appearing within $L^{y}_{B,t_2}$ (or vice versa).
Again, to obtain the $\PI_s$ ledgers $L^{x}_{s,t_1}$ and $L^{y}_{s,t_2}$, clients $x$ and $y$ iteratively run the sanitization step at Line~\ref{line:sanitization} of Alg.~\ref{alg.serial} on the snapshots $\snp_i$, $i = 1,2, \ldots$, and $\snp'_i$, $i = 1,2, \ldots$, respectively.
Then, by iteratively applying Proposition~\ref{prop:consistency-2} on these snapshots, we observe that the $\PI_s$ ledger $L^{x}_{s,t_1}$ obtained by client $x$ at time $t_1$ is consistent with the $\PI_s$ ledger $L^{y}_{s,t_2}$ obtained by client $y$ at time $t_2$.
Hence, it holds that for any clients $x$ and $y$ and times $t_1$ and $t_2$, the $\PI_s$ ledgers in the view of $x$ and $y$ at times $t_1$ and $t_2$ are consistent.

Finally, when both $\PI_A$ and $\PI_B$ are live with parameter $\Tconf$ after GST, if the environment inputs $\PI_s$ a transaction $\tx$ at time $t$, the transaction appears and stays in the certified $\PI_A$ ledger of an honest $\PI_A$ validator by time $\max(\GST, t) + \Tconf$ and is input to $\PI_B$.
Then, it appears and stays in the $\PI_B$ ledger of every client as part of a certified $\PI_A$ snapshot by time $\max(\GST, t) + 2\Tconf$.
Note that the first instance of $\tx$ cannot be sanitized while clients generate their $\PI_s$ ledgers from their $\PI_B$ ledgers.
Hence, $\tx$ appears and stays in the $\PI_s$ ledgers of all clients by time $\max(\GST, t) + 2\Tconf$.

When both $\PI_A$ and $\PI_B$ generate certificates, so does $\PI_s$; since the $\PI_s$ ledger is extracted from the certified $\PI_B$ ledger and the certified $\PI_A$ snapshots within that ledger.
Moreover, when both $\PI_A$ and $\PI_B$ proceed in epochs of fixed duration, so does $\PI_s$; as both the $\PI_B$ ledger and the snapshots therein grow in discrete time steps.
\end{proof}

\subsection{Proof of Theorem~\ref{thm:parallel-combined-security}}
\label{sec:proof-parallel}

When all blockchains generate certificates, so does $\PI_p$; since the overlay protocol generates certificates (by virtue of being secure under partial synchrony).
Moreover, when all blockchains proceed in epochs of fixed duration, so does $\PI_p$; as the overlay protocol proceeds in epochs of fixed duration.

Given the above observations for points (3) and (4) in Theorem~\ref{thm:parallel-combined-security}, proof of Theorem~\ref{thm:parallel-combined-security} follows from Theorems~\ref{thm:safe} and~\ref{thm:live}.
In the subsequent proofs, we only consider validators emulated by smart contracts on the finalized and valid $\PI_i$ ledgers observed by clients (\ie, emulated overlay validators).
When we talk about two separate actions by a validator, these actions might have appeared in different execution traces observed by different clients or the same client at different times.
We use the same notation as in Appendix~\ref{sec:oft-protocol-detailed} for the overlay protocol specific messages exchanged among the validators via the CCC abstraction.
All blocks referred below are blocks of the overlay protocol.

\begin{proposition}
\label{prop:safe-chain-1}
If a blockchain is safe, the validator emulated by the smart contract running on this blockchain does not send two conflicting vote messages $\lra{\text{vote}(B,v)}$ and $\lra{\text{vote}(B',v)}$ with the same epoch $v$, in the views of any (potentially the same) clients.
Similarly, no leader $L_v$ of any epoch $v$ proposes two conflicting blocks for $v$, in the views of any (potentially the same) clients.
\end{proposition}

This proposition directly follows from the safety of the $\PI_i$ ledgers.

\begin{proposition}
\label{prop:safe-chain-2}
No validator creates an invalid execution of the overlay protocol in the view of any client.
\end{proposition}

This proposition directly follows from the fact that the clients refuse to output ledgers with invalid executions according to the external validity rules.

\begin{lemma}
\label{lem:safe}
Suppose that every blockchain is safe. 
Then, no two validators commit two conflicting overlay blocks in the views of any (potentially the same) clients.
\end{lemma}
\begin{proof}
Suppose two validators $\val$ and $\val'$ commit two overlay blocks $B$ and $B'$ at epochs $v$ and $v' \geq v$ respectively in the views of two (potentially the same) clients.
We next show that $B \preceq B'$:
\begin{itemize}
    \item Suppose $v = v'$. Since the leader $L_v$ proposes a single block for its epoch by Proposition~\ref{prop:safe-chain-1}, $B = B'$.
    \item We show the $v' > v$ case by induction:

    \smallskip
    \noindent
        Suppose $v'=v+1$. Then, $\val$ must have received $f+1$ acknowledge messages $\lra{\text{ack}(B,v)}$ by unique validators for block $B$.
        If $L_{v+1}$ received a certificate $C_v(B)$ for block $B$ and epoch $v$ before entering epoch $v+1$, then by Proposition~\ref{prop:safe-chain-2}, the proposed block $B'$ extends $B$. 
        If $L_{v+1}$ received $f+1$ leader-down messages before entering epoch $v+1$, by quorum intersection, at least one of them must have been sent by a validator who sent an acknowledge message $\lra{\text{ack}(B,v)}$ for block $B$ and epoch $v$. 
        By Proposition~\ref{prop:safe-chain-1}, if a leader-down message contains a block from epoch $v$, that block must be $B$.
        Then, by Proposition~\ref{prop:safe-chain-2}, the leader $L_{v+1}$ proposes $B' = B_{k+1}$ extending $B$.
        
        \smallskip
        \noindent
        Suppose that for epochs less than $v'$, the block proposed by the leader extends or is equal to $B$. 
        For epoch $v'$, if $L_{v'}$ received a certificate $C_{v'-1}(B'')$ for some block $B''$ and epoch $v'-1$ before entering epoch $v'$, then by Proposition~\ref{prop:safe-chain-2}, the proposed block $B' = B_{k'}$ extends $B''$ and thus $B$.
        If $L_{v'}$ received $f+1$ leader-down messages for epoch $v'-1$ before entering epoch $v'$, by quorum intersection, at least one of them must have been sent by a validator who sent an acknowledge message $\lra{\text{ack}(B,v)}$ for block $B$ and epoch $v$.
        Thus, at least one of the leader-down messages must have included a block $B''$ from an epoch $v'' \geq v$. 
        By Proposition~\ref{prop:safe-chain-1}, if a leader-down message contains a block from epoch $v''$, that block must be $B''$, and without loss of generality, let $B''$ be the block with the highest epoch that is included among these leader-down messages.
        Then, by Proposition~\ref{prop:safe-chain-2}, the leader $L_{v'}$ proposes $B_{k'}$ extending $B''$.
        From the assumption, every block proposed by the leaders in epochs in $[v,v'-1]$ extends or is equal to $B$. 
        This implies that $B_{k'}$ extends $B$.
        Finally, since the leader $L_{v'}$ proposes a single block for epoch $v'$ by Proposition~\ref{prop:safe-chain-1}, $B' = B_{k'}$ extends $B$, \ie, $B \preceq B'$.
\end{itemize}
\end{proof}

Finally, we prove safety of the overlay protocol when every constituent blockchain is safe. 

\begin{theorem}\label{thm:safe}
Suppose that every blockchain is safe. Then, the parallel composition is safe.
\end{theorem}
\begin{proof}
Recall that a client accepts an overlay block and its prefix chain if it is part of the committed chains of $2$ or more validators.
By Lemma~\ref{lem:safe}, all blocks committed by all validators in the clients' views are consistent.
Hence, for any two overlay blocks $B$ and $B'$ accepted by the clients $\client$ and $\client'$ at times $t$ and $t'$, it holds that $B$ and $B'$ are consistent.
Moreover, once a client outputs a block and its chain, it does not ever output a shorter chain.
Thus, the parallel composition satisfies safety.
\end{proof}

We next prove liveness for the emulated validators.

\begin{lemma}\label{lem:live-validator}
Suppose that at least $2$ blockchains are $\Tconf$-live.
Then, all validators emulated on the live blockchains, \ie, live validators, commit a new block proposed by a live validator soon after GST.
\end{lemma}
\begin{proof}
Suppose the leader $L_v$ of an epoch $v$ starting at some time $t > \GST + \Tconf$ is a live validator.
By the leader selection rule, there exists such a first epoch with $t < \GST +  4 \Tconf$.

If $L_v$ observes a certificate $C_{v-1}(B_{k-1})$ for a block $B_{k-1}$ from epoch $v-1$ before entering epoch $v$, since it is live, it builds a block $B_k$ extending $B_{k-1}$ and broadcasts $B_{k-1}$ along with its ticket.
Otherwise, it must have received $f+1$ leader-down messages for epoch $v-1$, as epoch $v$ starts $\Tconf$ time after $\GST$.
Then, $L_v$ builds a block $B_k$ by extending the block $B_{k'}$ associated with the maximal epoch $v'$ among those contained by the $f+1$ leader-down messages, and broadcasts $B_{k}$ along with its ticket.
Hence, in either case, $L_v$ broadcasts a proposal and its ticket.

Upon receiving $L_v$'s proposal within $\Tconf$ time after its broadcast, all live validators broadcast an acknowledge messages $\lra{\text{ack}(B_k,v)}$ for the proposed block.
Therefore, by time $t + 2\Tconf$, all live validators receive a certificate for the live leader's block, committing it by time $t + 2\Tconf$.
Thus, all live validators commit a new block proposed by a live validator soon after GST.
\end{proof}

Finally, we prove liveness of the overlay protocol when $f+1$ of the constituent blockchains are live. 
\begin{theorem}\label{thm:live}
Suppose that at least $f+1$ blockchains are $\Tconf$-live.
Then, the parallel composition is live with constant latency.
\end{theorem}
\begin{proof}
Recall that a client accepts an overlay block and its chain if it is part of the committed chains of $f+1$ or more validators.
By Lemma~\ref{lem:live-validator},  all live validators, $f+1$ in total, commit a new block proposed by a live validator soon after GST, \ie, within $t = f(3\Tconf) + 3\Tconf$ time.
Hence, all clients observing the parallel composition after time $t$ accept an overlay block proposed by a live validator and its chain.
To ensure that this block persists in their ledger, even if it is conflicting with the past outputted ledgers, the clients add the block and its prefix chain into their ledgers and sanitize the duplicate transactions.
Since the overlay block proposed by a live validator contains all outstanding transactions received by honest validators of the underlay protocols, liveness is satisfied for clients with constant latency, \ie, $f(3\Tconf) + 3\Tconf$.
\end{proof}

\subsection{Proof of Theorem~\ref{lem:circuit-general-psync}}
\label{sec:app-circuit-general-psync}

\begin{proof}[Proof of Theorem~\ref{lem:circuit-general-psync}]
Fix a tuple $(E^S, E^L)$ as described by the theorem statement.
Let $A := A(E^L) = \{\ind(l) | (0^k, l) \in E^L\}$, and $Q_1, ..., Q_m$ denote the sets of indices within $A$. 
Here, $m$ is a finite number, bounded by $2^n$. 
We inductively build the overlay protocol characterized by the tuple $(E^S, E^L)$:

\noindent \textbf{Induction hypothesis:}
For any $(E^S,E^L)$ as described by the theorem such that $A(E^L)$ consists of $m$ sets (quorums) $Q_1, \ldots, Q_m$, there exists an overlay protocol characterized by a tuple dominating $(E^S,E^L)$. 

\noindent \textbf{Base case:} 
Suppose $m=1$, \ie, $A = \{Q_1\}$. 
Then, each $(s,l) \in E^S$ requires the safety of one of the protocols in $\{\PI_i | i\in Q_1\}$ by clause (2).
Therefore, the serial composition of the protocols $\{\PI_i | i\in Q_1\}$ is characterized by a tuple dominating $(E^S,E^L)$.
This follows from the iterative application of Theorem~\ref{thm:serial-security} on the protocols $\{\PI_i,i\in Q_1\}$.

\noindent \textbf{Inductive step}:
Suppose $A = \{Q_1, \ldots, Q_m,  Q_{m+1}\}$.
By the induction hypothesis, for any $(\tilde E^S,\tilde E^L)$ as described by the theorem statement such that $A := A(\tilde E^L) = \{Q_1, \ldots, Q_m\}$, there exists an overlay protocol $\PI$ characterized by a tuple dominating $(E^S,E^L)$.
Then, using $\PI$ and the \lvl composition, we construct the following protocol $\PI'$:
\begin{itemize}
\item The first protocol of the \lvl composition is $\PI$ itself.
\item The second protocol of the \lvl composition is the serial composition of the protocols $\{\PI_i, i\in Q_{m+1}\}$.
\item The third protocol of the \lvl composition is the $(2f+1)$-\lvl composition, where $f = m-1$, \ie, a $(2m-1)$-\lvl composition:

\noindent
(i) Of the $2m-1$ protocols composed, $m$ are as follows: For each $i \in [m]$, the $i$-th composition is the serial compositions of all protocols within $Q_{i,m+1} := Q_i \cap Q_{m+1}$. 

\noindent
(ii) Remaining $m-1$ of $2m-1$ protocols are copies of $\PI$.
\end{itemize}

First, we show that $\PI'$ is live if there exists $Q \in A$ such that $\PI_i$ is live for all $i \in Q$. 
Suppose there is a $j \in [m]$ such that all $\PI_i$, $i\in Q_j$, are live. 
Then, by the induction hypothesis, $\PI$ is live, \ie, the first protocol of the \lvl composition is live.
By Theorem~\ref{thm:serial-security}, all protocols within $Q_{j,m+1}$ are live. 
As at least $m$ components in the $(2m-1)$-\lvl composition are live, the third protocol of the \lvl composition is live by Lemma~\ref{lem:root} (which can be achieved by a circuit protocol as shown in Appendix~\ref{sec:parallel-3-3-2}).
Hence, by Theorem~\ref{thm:parallel-combined-security}, $\PI'$ is live.

Now, suppose all $\PI_i$, $i\in Q_{m+1}$, are live. 
By Theorem~\ref{thm:serial-security}, the second protocol of the \lvl composition is live. 
Again, by Theorem~\ref{thm:serial-security}, all protocols within $Q_{j,m+1}$ are live. 
Therefore, again, the third protocol of the \lvl composition is live by the definition of the $(2m-1)$-\lvl composition. 
By Theorem~\ref{thm:parallel-combined-security}, this implies that $\PI'$ is live.

We next prove that $\PI'$ is safe if for every pair $Q_1, Q_2 \in A \cup \{Q_{m+1}\}$, at least one chain $\PI_i$ for $i\in Q_1\cap Q_2$ is safe.
Note that the conditions on $\tilde E^S$ require $\PI$ to be safe if for every pair $Q_1, Q_2 \in A$, at least one chain $\PI_i$, $i \in Q_1 \cap Q_2$ is safe.
As the conditions for the safety of $\PI'$ imply those for $\PI$, if they hold, $\PI$ is safe.
Moreover, these conditions for the safety of $\PI'$ imply that at least one chain in $Q_{m+1}$ is safe, which by Theorem~\ref{thm:serial-security} implies that the second protocol of the \lvl composition is safe.
Finally, by the same conditions, for each $i \in [m]$, there exists a protocol that is safe within one of the sets $Q_{i,m+1} := Q_i \cap Q_{m+1}, i \in [m]$.
Thus, all of the $2m-1$ protocols within the $(2m-1)$-\lvl composition are safe, which implies that the third protocol of the \lvl composition is safe by Lemma~\ref{lem:root}.
Consequently, by Theorem~\ref{thm:parallel-combined-security}, $\PI'$ is safe.
Combining the safety and liveness results for $\PI'$, we can conclude that $\PI'$ is characterized by a tuple dominating $(E^S,E^L)$.
\end{proof}

\subsection{Proof of Theorem~\ref{lem:converse-symmetric-psync}}
\label{sec:app-converse-symmetric-psync}

\begin{proof}[Proof of Theorem~\ref{lem:converse-symmetric-psync}]
For contradiction, suppose there are tuples $(m_s, m_l, m_{sl}) \in p^L$, $(n_s, n_l, n_{sl}) \in P^S$ such that $m_l > k/2$, but $n_s \leq 2(k-m_l)$.
Let $l^1 = 1^{m_l}0^{k-m_l}$, $l^2 = 0^{k-m_l} 1^{m_l}$ and $s^3 = 1^{k-m_l}0^{2m_l-k}1^{k-m_l}$.
Since the number of $1$'s in $s^3$ is $2(k-m_l)$, by definition of permutation invariant protocols, there exist $s^1, s^2, l^3 \in \{0,1\}^k$ such that $(s^1,l^1)\in V^L, (s^2,l^2)\in V^L, (s^3,l^3)\in V^S$.
However, $\ind(l^1)\cap\ind(l^2)\cap\ind(s^3) = \emptyset$, which is a contradiction.

Now, suppose there exist a tuple $(m_s, m_l, m_{sl}) \in p^L$ such that $m_l \leq k/2$.
Let $l^1 = 1^{m_l}0^{k-m_l}$ and $l^2 = 0^{k-m_l} 1^{m_l}$.
Since the number of $1$'s in $l^3$ and $l^2$ are at most $k/2$, by definition of permutation invariant protocols, there exist $s^1, s^2 \in \{0,1\}^k$ such that $(s^1,l^1)\in V^L, (s^2,l^2)\in V^L$.
However, $\ind(l^1)\cap\ind(l^2) = \emptyset$, which is a contradiction.
\end{proof}

\subsection{Proof of Theorem~\ref{thm:lvs-security}}
\label{sec:lvs-security}

\begin{proof}[Proof of Theorem~\ref{thm:lvs-security}]
Without loss of generality, suppose $\PI_A$ is live with some parameter $\Tconf$.
Then, once a transaction $\tx$ is input to an honest validator at some time $t$, it appears and stays in the $\PI_A$ ledgers of all clients by time $t+\Tconf$.
If the interleaving condition is satisfied in the view of a client $\client$ observing the protocol at some time $t' \geq t+2\Tconf$, then $\tx \in L^\client_{B,t'}$.
Thus for any such client $\client$ and time $t' \geq t+2\Tconf$, this implies $\tx \in L^\client_{p,t'}$ per eq.~\eqref{eq:lvs-1}.
Therefore, $\tx$ appears (and stays) in the $\PI_p$ ledgers of all such clients by time $t+2\Tconf$.
If the interleaving condition is not satisfied in the view of a client $\client$ observing the protocol at some time $t' \geq t+2\Tconf$, then it is still the case that $\tx \in L^\client_{p,t'}$, since $L^\client_{p,t'}$ would then contain every transaction in $L^\client_{A,t'-\Tconf}$ per eq.~\eqref{eq:lvs-2}.
This proves the liveness clause.

Suppose both $\PI_A$ and $\PI_B$ are safe and live with parameter $\Tconf$.
Then, for any client $\client$ and time $t$, it holds that any transaction that appears in the ledger $L^\client_{A,t}$ or $L^\client_{B,t}$ is in the $\PI_A$ and $\PI_B$ ledgers of all clients by time $t+\Tconf$.
Therefore the interleaving condition is satisfied in the view of all clients.
Then, for any client $\client$ and time $t$, $L^\client_{p,t}$ is as defined by eq.~\ref{eq:lvs-1}, which is an interleaving of two ledgers.
Moreover, by the safety of the chains $\PI_A$ and $\PI_B$, for all positive integers $\ell$ and $\ell'$, clients $\client$ and $\client'$, and times $t$ and $t'$, the ledgers $\textsc{Interleave}(L^\client_{A,t}[{:}\ell],L^\client_{B,t}[{:}\ell])$ and $\textsc{Interleave}(L^{\client'}_{A,t'}[{:}\ell'],L^{\client'}_{B,t'}[{:}\ell'])$ are consistent with each other.
This proves the safety clause.

Finally, if both $\PI_A$ and $\PI_B$ generate certificates, since the $\PI_p$ ledger is a deterministic function of the $\PI_A$ and $\PI_B$ ledgers when $\PI_p$ is safe, the tuple of certificates for the $\PI_A$ and $\PI_B$ ledgers act as a certificate for the $\PI_p$ ledger. 
Thus, $\PI_p$ generates certificates.
If both $\PI_A$ and $\PI_B$ proceed in epochs of fixed duration, chain growth rate of the $\PI_p$ ledger is bounded as well, implying that it proceeds in epochs of fixed duration.
\end{proof}

\subsection{Proof of Lemma~\ref{lem:root}}
\label{sec:parallel-3-3-2}

Proof relies on Lemma~\ref{lem:leave} from Section~\ref{sec:circuit-symmetric-psync}.

\begin{proof}[Proof of Lemma~\ref{lem:root}]
We prove by mathematical induction. 
Suppose that for any $g<f$, $(2g+1)$-\lvl is achievable by using $\pi^{(2g+1)}$. 
Note that the base case is $g=1$, and we assume that there exists a protocol achieving the $(3,3,2)$ point.

By Lemma~\ref{lem:leave}, we note that $\pi^{(2g+2)}$ is also achievable. 
Let $\PI_1,\dots,\PI_{2f+1}$ be $2f+1$ different blockchains. Consider all subsets of $\PI_1,\dots,\PI_{2f+1}$ with $2f-1$ blockchains, which are $f(2f+1)$ subsets in total. We use $\pi^{(2f-1)}$ for these subsets and obtain $f(2f+1)$ blockchains, say $\PI_i^1$ with $i=1,\dots,f(2f+1)$.  %

If $f$ is an odd number, then we split $f(2f+1)$ blockchains into $f+2$ groups, where $f+1$ groups has $2f-1$ blockchains and the remaining group has $1$ blockchain. Then, we run $\pi^{(2f-1)}$ on the first $f+1$ groups and obtain $f+1$ output blockchains. Finally, we run $\pi^{(f+2)}$ to obtain the final output blockchain.

If $f$ is an even number, then we split $f(f+1)$ blockchains into $f+1$ groups, where $f$ groups have $2f-1$ blockchains and the remaining group has $2f$ blockchain. Then, we run $\pi^{(2f-1)}$ on the first $f$ groups and run $\pi^{(2f,2f-1,f+1)}$ on the last group. Hence, we obtain $f+1$ output blockchains. Finally, we run $\pi^{(f+1)}$ to obtain the final output blockchain.

It is easy to show that if $\PI_1,\dots,\PI_{2f+1}$ are safe then the protocol is safe. Assuming that at least $f+1$ of $\PI_1,\dots,\PI_{2f+1}$ are live, we will show that the protocol is live. 
We prove by contradiction. 
Suppose that the protocol is not live. First, we note that as at least $f+1$ of $\PI_1,\dots,\PI_{2f+1}$ are live, at most $f(f+1)/2$ blockchains among $\{\PI_i^1\}_{i=1}^{f(f+1)}$ output by $\pi^{(2f-1)}$ running on size $2f-1$ subsets are not live. 
This is because to obtain a not-live blockchain $\PI_{i}^1$ as output of $\pi^{(2f-1)}$, at least $f-1$ of input blockchains need to be not-live, which implies that this subset with size $2f-1$ needs to drop $2$ live blockchains from $\PI_1,\dots,\PI_{2f+1}$. 

Now, if $f$ is an odd number, as the protocol is not live, at least $(f+3)/2$ blockchains as input of $\pi^{(f+2)}$ are not live. 
This implies that at least $f(f+1)/2+1$ blockchains among $\{\PI_i^1\}_{i=1}^{f(2f+1)}$ are not live. As $f(f+1)/2+1>f(f+1)/2$, this leads to a contradiction. 
On the other hand, if $f$ is an even number, as the protocol is not live, at least $(f+2)/2$ blockchains as input of $\pi^{(f+1)}$ are not live. This implies that at least $(f+2)f/2$ blockchains among $\{\PI_i^1\}_{i=1}^{f(2f+1)}$ are not live. 
As $(f+2)f/2>(f+1)f/2$, this leads to a contradiction. 
\end{proof}

\subsection{Alternative Proof for Theorem~\ref{thm:informal}}
\label{sec:alternative}

The following is an analogue of Theorem~\ref{lem:circuit-symmetric-sync} for partial synchrony and stated here for completeness:

\begin{theorem}
\label{lem:circuit-symmetric-psync}
If a protocol is characterized by $(P^S,P^L)$ in the set
$$
\{(\{(2(k-m_l)+1,0,0)\},\{(0,m_l,0)\})|k/2<m_l\leq k\},
$$
then there exists a permutation invariant overlay protocol characterized by a tuple dominating $(P^S,P^L)$ under partial synchrony.
\end{theorem}

Theorem~\ref{lem:circuit-symmetric-psync} follows as a corollary of Theorem~\ref{lem:circuit-general-psync}.

\begin{proof}
Fix some $m_l$, $k/2<m_l\leq k$.
Consider the tuple $P = (P^S,P^L) = (\{(n_s = 2(k-m_l)+1, n_l = 0, n_{sl} = 0)\},\{(m_s = 0, m_l, m_{sl} = 0)\})$.
Let $(V^S,V^L) = V(P)$, $E^S = \text{exm}(V^S)$ and $E^L = \text{exm}(V^L)$.
\begin{enumerate}
    \item Since $n_l = n_{sl} = m_s = m_{sl} = 0$, for all $(s,l)\in E^L$, $(s',l') \in E^S$, it holds that $s = l' = 0^k$.
    \item Since $n_s = 2(k-m_l)+1$, by quorum intersection, for all $(0^k,l^1), (0^k,l^2)\in V^L, (s,0^k)\in E^S$, it holds that $\ind(l^1) \cap \ind(l^2) \cap \ind(s) \neq \emptyset$.
\end{enumerate}
Therefore, by Theorem~\ref{lem:circuit-symmetric-psync}, there exists an overlay protocol characterized by a tuple dominating $(P^S,P^L)$.
Moreover, the overlay protocol suggested by the construction in the proof of Theorem~\ref{lem:circuit-symmetric-psync} is permutation invariant.
\end{proof}

\subsection{Proof of Theorem~\ref{lem:circuit-general-sync}}
\label{sec:appendix-circuit-general-sync}

\begin{proof}[Proof of Theorem~\ref{lem:circuit-general-sync}]
Fix a tuple $(E^S, E^L)$ as described by the theorem statement.
Let $A := A(E^L) = \{\ind(l) | (0^k, l) \in E^L\}$, and $Q_1, ..., Q_m$ denote the sets of indices within $A$. 
Here, $m$ is a finite number, bounded by $2^n$. 
We inductively build the overlay protocol characterized by the tuple $(E^S, E^L)$:

\noindent \textbf{Induction hypothesis:}
For any $(E^S,E^L)$ as described by the theorem statement such that $A(E^L)$ consists of $m$ sets (quorums) $Q_1, \ldots, Q_m$, there exists an overlay protocol characterized by a tuple dominating $(E^S,E^L)$. 

\noindent \textbf{Base case:} 
Suppose $m=1$, \ie, $A = \{Q_1\}$. 
Then, each $(s,l) \in E^S$ requires the safety of one of the protocols in $\{\PI_i | i\in Q_1\}$.
Therefore, the serial composition of the protocols $\{\PI_i | i\in Q_1\}$ is characterized by a tuple dominating $(E^S,E^L)$.
This follows from the iterative application of Theorem~\ref{thm:serial-security} on the protocols $\{\PI_i,i\in Q_1\}$.

\noindent \textbf{Inductive step}:
Suppose $A = \{Q_1, \ldots, Q_m,  Q_{m+1}\}$.
By the induction hypothesis, for any $(\tilde E^S,\tilde E^L)$ as described by the theorem statement such that $A := A(E^L) = \{Q_1, \ldots, Q_m\}$, there exists an overlay protocol $\PI$ that dominates the protocol characterized by $(E^S,E^L)$.
Then, using $\PI$, we construct the following protocol $\PI'$ using the \lvl composition:
\begin{itemize}
\item The first protocol of the \lvl composition is $\PI$ itself.
\item The second protocol of the \lvl composition is the serial composition of the protocols $\{\PI_i, i\in Q_{m+1}\}$.
\item The third protocol of the \lvl composition is the $(2f+1)$-\lvl composition, where $f = m-1$, \ie, a $(2m-1,2m-1,m-1)$ composition:

\noindent
(i) Of the $2m-1$ protocols composed, $m$ are as follows: For each $i \in [m]$, the $i$-th composition is the serial compositions of all protocols within $Q_{i,m+1} := Q_i \cap Q_{m+1}$ and 
the \lvs compositions of all tuples in $\{(\Pi_{j_1}, \PI_{j_2}) | j_1\in Q_i, j_2\in Q_{m+1}\}$).

\noindent
(ii) Remaining $m-1$ of $2m-1$ protocols are copies of $\PI$.
\end{itemize}

First, we show that $\PI'$ is live if there exists $Q \in A$ such that $\PI_i$ is live for all $i \in Q$. 
Suppose there is a $j \in [m]$ such that all $\PI_i$, $i\in Q_j$, are live. 
Then, by the induction hypothesis, $\PI$ is live, \ie, the first protocol of the \lvl composition is live.
By Theorem~\ref{thm:lvs-security}, all protocols within $\{\text{\lvs}(\Pi_{j_1}, \PI_{j_2}) | j_1\in Q_j, j_2\in Q_{m+1}\}$ are live.
By Theorem~\ref{thm:serial-security}, all protocols within $Q_{j,m+1}$ are live. 
As at least $m$ components in the $(2m-1,2m-1,m-1)$ composition are live, the third protocol of the \lvl composition is live by the definition of the $(2f+1)$-\lvl composition (which can be achieved by a circuit protocol as shown in Appendix~\ref{sec:parallel-3-3-2}).
Hence, by Theorem~\ref{thm:parallel-combined-security}, $\PI'$ is live.

Now, suppose all $\PI_i$, $i\in Q_{m+1}$, are live. 
By Theorem~\ref{thm:serial-security}, the second protocol of the \lvl composition is live. 
Again, by Theorem~\ref{thm:lvs-security}, all protocols within $\{\text{\lvs}(\Pi_{j_1}, \PI_{j_2}) | j_1\in Q_j, j_2\in Q_{m+1}\}$ are live, and by Theorem~\ref{thm:serial-security}, all protocols within $Q_{j,m+1}$ are live. 
Therefore, again, the third protocol of the \lvl composition is live by the definition of the $(2f+1)$-\lvl composition. 
By Theorem~\ref{thm:parallel-combined-security}, this implies that $\PI'$ is live.

We next prove that $\PI'$ is safe if for every pair $Q_1, Q_2 \in A \cup \{Q_{m+1}\}$, either at least one chain $\PI_i$ for $i\in Q_1\cap Q_2$ is safe or at least two chains $\PI_{j_1},\PI_{j_2}$, $j_1\in Q_1$, $j_2\in Q_2$ are safe and live. 
Note that the conditions on $\tilde E^S$ require $\PI$ to be safe if for every pair $Q_1, Q_2 \in A$, either at least one chain $\PI_i$, $i \in Q_1 \cap Q_2$ is safe, or at least two chains $\PI_{j_1},\PI_{j_2}$, $j_1\in Q_1$, $j_2\in Q_2$ are safe and live.
As the conditions for the safety of $\PI'$ imply those for $\PI$, if they hold, $\PI$ is safe.
Moreover, these conditions for the safety of $\PI'$ imply that at least one chain in $Q_{m+1}$ is safe, which by Theorem~\ref{thm:serial-security} implies that the second protocol of the \lvl composition is safe.
Finally, by the same conditions, for each $i \in [m]$, there exists a protocol that is safe within one of the sets $Q_{i,m+1} := Q_i \cap Q_{m+1}, i \in [m]$ or $\{\text{\lvs}(\Pi_{j_1}, \PI_{j_2}) | j_1\in Q_i, j_2\in Q_{m+1}\}$ by Theorem~\ref{thm:lvs-security}.
Thus, all of the $2m-1$ protocola within the $(2m-1, m-1, m-1)$ composition are safe, which implies that the third protocol of the \lvl composition is safe by the definition of the $(2f+1)$-\lvl composition.
Consequently, by Theorem~\ref{thm:parallel-combined-security}, $\PI'$ is safe.
Combining the safety and liveness results for $\PI'$, we can conclude that $\PI'$ is characterized by a tuple dominating $(E^S,E^L)$.
\end{proof}

\subsection{Proof of Theorem~\ref{lem:converse-general-sync}}
\label{sec:app-converse-general-sync}

\begin{proof}[Proof of Theorem~\ref{lem:converse-general-sync}]
For contradiction, suppose $\ind(l^1)\cap\ind(l^2)\cap\ind(s^3) = \emptyset$ and without loss of generality, there is no index $i \in \ind(l^1)$ such that $s_i = l_i = 1$.
Denote the $k$ underlay blockchains by $\PI_1,\dots,\PI_k$. 
There are two clients $\client_1, \client_2$.
Consider the following three worlds. 

\noindent 
\textbf{World 1:} 
All blockchains are safe.
The underlay chains $\PI_i$, $i \in \ind(l^1)$, are live.
The chains $\PI_i$, $i \in \ind(l^2)$ are not live, and the rest are stalled.
Suppose transactions $\tx_1$ and $\tx_2$ are input to the protocol at time $t = 0$.
The chains $\PI_i$, $i \in \ind(l^2)$, start with the transactions and output ledgers in $\client_1$'s view.
However, they seem stalled to the validators of the chains in some set $P_l \subseteq \{\PI_i | i \in \ind(l^1) \setminus \ind(l^2)\}$.
As $f_L^{\PI}(s^1,l^1)=1$, the overlay blockchain is live. 
At time $t_1=\Tconf$, the client $\client_1$ outputs its interchain ledger containing $\tx_1$ and $\tx_2$.
Without loss of generality, suppose $\tx_1$ appears as the first entry:
$L_{t_1}^{\client_1}=[\tx_1, \tx_2]$.

\noindent
\textbf{World 2:} 
All blockchains are safe.
The underlay chains $\PI_i$, $i \in \ind(l^2)$, are live.
The chains in the set $P_l$ are not live, but those in $\{\PI_i | i \in \ind(l^1)\} \setminus P_l$ are live.
The rest are stalled.

Suppose that $\tx_1$ and $\tx_2$ are input to the protocol at time $t = 0$.
The chains in $P_l$ seem stalled to $\client_2$ and the validators of the chains $\PI_i$, $i \in \ind(l^2) \setminus \ind(l^1)$.
As $f_L^{\PI}(s^2,l^2)=1$, the overlay blockchain is live. 
At time $t_1=\Tconf$, $\client_2$ outputs its interchain ledger containing $\tx_1$ and $\tx_2$.
Suppose $\tx_2$ appears as the first entry:
$L_{t_1}^{\client_2}=[\tx_2, \tx_1]$.

\noindent\textbf{World 3:}
None of the chains in $\ind(l^1)$ is both safe and live.
Here, $P_s$ and $P_l$ denote the chains $\PI_i$, $i \in \ind(l^1)$, that are \emph{not} safe and live respectively.
The chains $\PI_i$, $i\in \ind(l^1) \cap \ind(l^2)$ are not safe (but live), and $\PI_i$, $i \in \ind(l^2) \setminus \ind(l^1)$, are both safe and live. 
For simplicity, let $Q = \ind(l^1)\cap\ind(l^2)$.

Suppose that $\tx_1,\tx_2$ are input to the protocol at time $t=0$.
The chains in $P_l$ seem stalled to $\client_2$ and the validators of $\PI_i$, $i\in \ind(l^2)/Q$ until at least time $\Tconf$, and pretend like the chains $\PI_i$, $i\in \ind(l^2)/Q$ are stalled.
The chains in $P_s$ expose two conflicting ledgers, the ledgers in world 1 to client $\client_1$ and the validators of $\PI_i$ for $i\in \ind(l^1)/Q$, and the ledgers in world 2 to client $\client_2$ and the validators of $\PI_i$ for $i\in \ind(l^2)/Q$.
As the chains $\PI_i$, $i\in Q$ are not safe, they simultaneously interact with $\client_1$ and the chains $\PI_i$, $i\in \ind(l^1)/Q$ as in World 1 and with $\client_2$ and the chains $\PI_i$, $i\in \ind(l^2)/Q$ as in World 2. 
Then, as the chains in $P_s$ cover all chains $\PI_i$, $i\in \ind(l^1)/Q$ that are live, client $\client_1$ cannot distinguish World 1 and World 3 before $\Tconf$, which implies that $L^{\client_1}_{t_1} = [\tx_1, \tx_2]$.
Similarly, client $\client_2$ cannot distinguish World 2 and World 3, which implies that $L^{\client_2}_{t_2} = [\tx_2, \tx_1]$. 
However, $L^{\client_1}_{t_1}$ and $L^{\client_2}_{t_2}$ conflict with each other, which violates the safety of the overlay protocol.
This is a contradiction.
\end{proof}

\subsection{Proof of Theorem~\ref{lem:converse-symmetric-sync}}
\label{sec:app-converse-symmetric-sync}

\begin{proof}[Proof of Theorem~\ref{lem:converse-symmetric-sync}]
Suppose there are tuples $(m_s, m_l, m_{sl}) \in p^L, (n_s, n_l, n_{sl}) \in P^S$ such that $m_l > k/2$, but $n_s \leq 2(k-m_l)$ and $n_{sl} \leq k - m_l$ for contradiction.
Let $l^1 = 1^{m_l}0^{k-m_l}$, $l^2 = 0^{k-m_l} 1^{m_l}$, $s^3 = 1^{k-m_l}0^{2m_l-k}1^{k-m_l}$ and $l^3 = 1^{k-m_l}0^{2m_l-k}1^{k-m_l}$.
Since the number of $1$'s in $s^3$ is $2(k-m_l)$ and the number of indices that are $1$ in both $s^3$ and $l^3$ is $k-m_l$, by definition of permutation invariant protocols, there exist $s^1, s^2 \in \{0,1\}^k$ such that $(s^1,l^1)\in V^L, (s^2,l^2)\in V^L, (s^3,l^3)\in V^S$.
However, in this case, $\ind(l^1)\cap\ind(l^2)\cap\ind(s^3) = \emptyset$, and there is no index $j \in \ind(l^2)$ such that $s^3_j = 1$ and $l^3_j = 1$, which is a contradiction.

Now, suppose there exist a tuple $(m_s, m_l, m_{sl}) \in p^L$ such that $m_l \leq k/2$ and $n_{sl} \leq k - m_l$.
Let $l^1 = 1^{m_l}0^{k-m_l}$, $l^2 = 0^{k-m_l} 1^{m_l}$, $s^3 = 1^k$ and $l^3 = 1^{m_l}0^{k-m_l}$
Since the number of $1$'s in $l^3$ and $l^2$ are at most $k/2$, by definition of permutation invariant protocols, there exist $s^1, s^2 \in \{0,1\}^k$ such that $(s^1,l^1)\in V^L, (s^2,l^2)\in V^L, (s^3,l^3) \in V^S$.
However, in this case, $\ind(l^1)\cap\ind(l^2) = \emptyset$, and there is no index $j \in \ind(l^2)$ such that $s^3_j = 1$ and $l^3_j = 1$, which is a contradiction.
\end{proof}

\subsection{Alternative Proof for Theorem~\ref{lem:circuit-symmetric-sync}}

Here, we give an alternative proof for Theorem~\ref{lem:circuit-symmetric-sync} using Sync Streamlet as an overlay protocol in the fashion of Trustboost~\cite{trustboost}.
Consider synchronous Streamlet~\cite{streamlet} with the quorom size $q$ and denote it as $\PI$.
We claim that $\PI$ is safe if $n_s \geq 2(k - q)+1$ or $n_{sl} \geq k-q+1$.

\begin{lemma}
Suppose $n_s \geq 2(k - q)+1$ or $n_{sl} \geq k-q+1$. 
If two blocks $B_1,B_2$ get notarized in consecutive epoch $e,e+1$ and $B_1$ is with length $l$, then, no block with length $l$ and epoch greater than $e+2$ is notarized in honest view.
\end{lemma}
\begin{proof}
Suppose that a block $B'$ with length $l$ and epoch $e'$ is notarized where $e'\geq e+3$. 
Let $Q_1$ be the set of blockchains which vote for $B_2$ in epoch $e+1$, let $Q_2$ be the set of blockchains which vote for $B'$ in epoch $e'$ and let $S$ be the set of blockchain which is safe and live. 

If we have $n_s+n_{sl}>2(k-q)$, as $|Q_1\cap Q_2|\geq 2q-n$, this implies that $k-n_s-n_{sl}<2q-k$. 
Therefore, at least one blockchain in $Q_1\cap Q_2$ is safe. 
Let $\PI_i$ be a safe blockchain in $Q_1\cap Q_2$. 
We note that $\PI_i$ has observed $B_1$ in epoch $e+1$. 
As the length of $B'$ is $l$, the voting action for $B'$ from $\PI_i$ is invalid. This contradicts with Proposition \ref{prop:safe-chain-2}.

If we have $n_{sl}>k-q$, then there exists a safe and live blockchain $\PI_i$ that vote for $B'$ in epoch $e'$. 
As $B_2$ get notarized in epoch $e+1$, $\PI_i$ has observed $B_2$ in the beginning of epoch $e+3$. 
As the length of $B'$ is less than the length of $B_2$, the voting action for $B'$ from $\PI_i$ is invalid. 
This contradicts with Proposition \ref{prop:safe-chain-2}.

\end{proof}

\begin{theorem}[Consistency]
Suppose that $\max\{s+b_s,2b_s\}>2(n-q)$. Suppose that two notarized chains $\texttt{chain}$ and $\texttt{chain'}$ of lengths $l$ and $l'$ respectively both triggered the finalization rule in some honest node’s view, that is, $\texttt{chain}[:l-5]$ and $\texttt{chain'}[:l'-5]$ are finalized in honest view. Without loss of generality, suppose that $|\texttt{chain}| \leq |\texttt{chain'}|$. It must be that $\texttt{chain}[:l-5]\preceq \texttt{chain'}[:l'-5]$
\end{theorem}

\end{document}